\documentclass[preprint,12pt]{elsarticle}

\usepackage{amssymb}
\usepackage{MnSymbol}
\usepackage{xspace}
\usepackage{multirow}
\usepackage{adjustbox}
\usepackage{amsthm}
\usepackage{comment}
\usepackage{arydshln}
\usepackage{subfigure}
\usepackage{changebar}

\usepackage{caption}
\usepackage{tikz}
\usepackage{lscape}
\usetikzlibrary{arrows,automata,shapes,snakes,backgrounds,petri,mindmap,fit,positioning}

\newcommand{\lan}[1]{\ensuremath{\mathbf{#1}}\xspace}

\newcommand{\stratstyle}[1]{\ensuremath{\mathrm{#1}}}

\newcommand{\ATL}[1][]{\lan{ATL_{\stratstyle{#1}}}}

\newcommand{\CSL}[1][]{\lan{CSL}}
\newcommand{\CSLP}[1][]{\lan{CSLP}}

\newcommand{\MIATL}[1][]{IATL[R]}

\newcommand{\ATLES}[1][]{\lan{ATLES}}
\newcommand{\ATELA}[1][]{\lan{ATELA}}

\newcommand{\coop}[2][]{\langle\!\langle{#2}\rangle\!\rangle_{_{\!\mathit{#1}}}}

\newcommand{\Next}[1][]{\!\raisebox{-.2ex}{ \mbox{\unitlength=0.9ex
            \begin{picture}(2,2)
            \linethickness{0.06ex}
            \put(1,1){\circle{2}}   \end{picture}}}_{{#1}}  \,}
\newcommand{\Sometm}[1][]{\Diamond_{{#1}}}
\newcommand{\Always}[1][]{\Box_{{#1}}}

\renewcommand{\Next}{\mathrm{X}}
\renewcommand{\Sometm}{\mathrm{F}}
\renewcommand{\Always}{\mathrm{G}}

\newcommand{\TransArrow}[1][]{\hookrightarrow....}

                                    \newcommand{\onepath}[1][]{\ensuremath{\lambda\ifthenelse{\equal{#1}{}}{}{[#1]}}}

\newcommand{\prop}[1]{\ensuremath{\mathsf{{#1}}}}

\newcommand{\plaus}[1][]{\ifthenelse{\equal{#1}{}}{\mathbf{P\;\!\!l}\,}{\mathbf{P\;\!\!l}_{#1}\,}}
\newcommand{\phys}[1][]{\ifthenelse{\equal{#1}{}}{\mathbf{P\;\!\!h}\,}{\mathbf{P\;\!\!h}_{#1}\,}}

\newcommand{\plaumodels}[1][]{\ensuremath{\ifthenelse{\equal{#1}{}}{\models_\sPlaupaths}{\models_{#1}}}}
\newcommand{\sPlaupaths}{\ensuremath{P}}

\usepackage{color}

\definecolor{lightgrey}{rgb}{0.8,0.8,0.8}
\definecolor{grey}{rgb}{0.6,0.6,0.6}
\definecolor{darkgrey}{rgb}{0.4,0.4,0.4}
\definecolor{darkgreen}{rgb}{0,0.7,0}

\newcommand{\bend}{20}

\newcommand{\onlabel}[1]{\colorbox{white}{\textit{{#1}}}}

\newcommand{\complexityclass}[1]{\ensuremath{\mathbf{{#1}}}\xspace}

\newcommand{\Ptime}{\complexityclass{P}}
\newcommand{\NP}{\complexityclass{NP}}

\newcommand{\Sigmacomplx}[1]{\complexityclass{\Sigma_{{#1}}^{\Ptime}}}
\newcommand{\Picomplx}[1]{\complexityclass{\Pi_{{#1}}^{\Ptime}}}

\newcommand{\putaway}[1]{}

\newcommand{\para}[1]{\smallskip\noindent\textbf{#1}}

\newenvironment{enumerate2}{\begin{enumerate}\itemsep 0in}{\end{enumerate}}

\newcommand{\finis}{{\scriptsize $\blacksquare$}}
\newcommand{\finisdef}{$\Box$}
\newcommand{\bul}{{\tiny $\blacksquare$}}

\def\itemiremember{\labelitemi}
\def\itemiiremember{\labelitemii}

\renewcommand{\leq}{\leqslant}
\renewcommand{\geq}{\geqslant}
\renewcommand{\phi}{\varphi}

\newcommand{\altbisim}[1]{\Lleftarrow\!\!\!\Rrightarrow_{#1}}
\newcommand{\notaltbisim}[1]{ \Lleftarrow\!\!\not\!\!\!\Rrightarrow_{#1}}

\newcommand{\br}[1]{\overline{#1}}
\newcommand{\ThreeBallot}{\textsf{ThreeBallot}\xspace}

\newtheorem{theorem}{Theorem}
\newtheorem{definition}[theorem]{Definition}
\newtheorem{lemma}[theorem]{Lemma}
\newtheorem{proposition}[theorem]{Proposition}
\newtheorem{corollary}[theorem]{Corollary}

\newtheorem{remark}[theorem]{Remark}
\newtheorem{example}{Example}

\newcommand{\G}{\mathcal G}

\newcommand\altsim[1]{\Rrightarrow_{#1}}

\newcommand{\mkCEGS}{\ensuremath{\text{iCGS}}\xspace}

\newcommand{\mkmrule}{\;|\;}
\newcommand{\mkbis}[1]{\thickapprox_{#1}}

\newcommand{\mkbiss}[1]{\mkbis{#1}^{subj}}
\newcommand{\mkbiso}[1]{\mkbis{#1}^{obj}}
\newcommand{\mkbisx}[1]{\mkbis{#1}^{x}}
\newcommand{\mkbissa}{\mkbiss{A}}
\newcommand{\nmkbissa}{\not \thickapprox_{A}^{subj}}
\newcommand{\nmkbissag}{\not \thickapprox_{Ag}^{subj}}
\newcommand{\mkbisoa}{\mkbiso{A}}
\newcommand{\mkbisxa}{\mkbisx{A}}
\newcommand{\mkpaltsima}{\paltsim{A}}
\newcommand{\mkpaltbisima}{\paltbisim{A}}

\newcommand{\mkstate}{q}
\newcommand{\mkstateb}{r}
\newcommand{\mkstrat}[1]{\sigma_{#1}}
\newcommand{\mkstrata}{\mkstrat{A}}

\newcommand{\mkimg}[2]{\mathit{Img}({#1}, {#2})}
\newcommand{\mkpartstrat}[2]{\mathit{PStr}_{#2}(#1)}
\newcommand{\mkpartstrata}[1]{\mkpartstrat{#1}{A}}

\newcommand{\mkmodelss}{\models_{\mathit{subj}}}

\newcommand\paltsim[1]{\rightsquigarrow_{#1}}
\newcommand\paltbisim[1]{\leftrightsquigarrow_{#1}}
\newcommand{\faltbisim[1]}{\Leftarrow\!\!\!\Rightarrow_{#1}}

\journal{Information \& Computation}

\begin{document}

\begin{frontmatter}

\title{Bisimulations for Verifying Strategic Abilities
  \\ with an Application to the ThreeBallot Voting Protocol}

\author[London]{Francesco Belardinelli}
\address[London]{Imperial College London, UK, and Universit\'e d'Evry, France}

\author[Creteil]{Rodica Condurache}
\address[Creteil]{LACL, Universit\'e Paris-Est Cr\'{e}teil, France}

\author[Creteil]{C\u{a}t\u{a}lin Dima}

\author[Warsaw]{Wojciech Jamroga}
\address[Warsaw]{Institute of Computer Science, Polish Academy of Sciences, Poland}

\author[Warsaw]{Michal Knapik}

\begin{abstract}
We propose a notion of alternating bisimulation for
strategic abilities under imperfect information. The bisimulation
preserves formulas of ATL$^*$ for both the {\em objective} and {\em
  subjective} variants of the state-based semantics with imperfect
information, which are commonly used in the modeling and verification
of multi-agent systems. Furthermore, we apply the theoretical result
to the verification of coercion-resistance in the ThreeBallot voting
system, a voting protocol that does not use cryptography. In
particular, we show that natural simplifications of an initial model
of the protocol are in fact bisimulations of the original model, and
therefore satisfy the same ATL$^*$ properties, including
coercion-resistance. These simplifications allow the model-checking
tool MCMAS to terminate on models with a larger number of voters and
candidates, compared with the initial model.
\end{abstract}


\end{frontmatter}

\section{Introduction}

\noindent

Formal languages for expressing strategic abilities of rational agents
have witnessed a steady growth in recent years
\cite{BullingDixJamroga10a,BullingJamroga14,Goranko+04a}.  Among the
most significant contributions we mention Alternating-time Temporal
Logic~\cite{AlurHenzingerKupferman97,AlurHenzingerKupferman02},
(possibly enriched with strategy contexts \cite{LaroussinieM15}),
Strategy Logic~\cite{ChatterjeeHenzingerPitterman07,MogaveroMPV14},
and Coalition Logic \cite{Pauly02}. These languages allow to express
that a group of agents has a strategy to enforce a certain outcome,
regardless of the behavior of the other agents.
That provides syntactical and semantic means to characterize winning
conditions in multi-player games, notions of equilibrium (e.g. Nash),
strategy-proofness, and so
on~\cite{ChatterjeeHenzingerPitterman07,Ball06abstraction,Hoek+05a,MogaveroMuranoVardi10a}.

However, if logics for strategies are to be applied to the
specification and verification of multi-agent systems (MAS)
\cite{GammieMeyden04a,Kacprzak+07a,LomuscioQuRaimondi15}, they need to
be coupled with efficient model checking techniques. Unfortunately,
while in contexts of perfect information we benefit from tractable
algorithms for model checking
\cite{AlurHenzingerKupferman02}, the situation is rather different
once we consider imperfect information. In contexts of imperfect
information the complexity of the verification task ranges between
$\Delta^P_2$-completeness~\cite{JamrogaDix06} to
undecidability~\cite{DimaT11}, depending on whether we assume perfect
recall. In this setting, complementary model checking techniques are
being investigated, in order to make the problem feasible in practice,
including semantic \cite{BLMR17,BLMRijcai17} and syntactic
restrictions \cite{CermakLMM18}, as well as
approximations \cite{JamrogaKK17}.

Amongst these, model reduction seems one of the more promising paths.
In particular, reductions by state- and action-space abstraction have
proved to be a valuable tool for efficient
verification \cite{Cousot77abstraction,ClarkeGrumbergLong94,Clarke+00b},
also in the context of strategic
abilities~\cite{Ball06abstraction,Alfaro04three,Lomuscio16abstraction-atlk,Belardinelli17abstractions}. Within
that approach, the ``concrete'' system $S$ to be model-checked is
abstracted into a ``simpler'' model $S^A$, which typically contains
less states and transitions and therefore can be easier to
verify. Then, the verification result is transferred from abstraction
$S^A$ to the concrete $S$ by virtue of some preservation result.
Normally, preservation is guaranteed by proving that abstraction $S^A$
is \emph{similar} or \emph{bisimilar} to $S$. (Bi)simulations are a
powerful tool to analyze the expressiveness of modal languages,
starting with van Benthem's characterisation of modal logic as the
bisimulation-invariant fragment of first-order
logic \cite{Blackburn+01a}. However, (bi)simulations are a lot less
understood in logics for strategies, where they have been studied
mostly in the contexts of perfect information
scenarios~\cite{AlurHKV98,Goltz92equivalences,Agotnes07irrevocable}.

In this paper, we advance the state-of-the-art by introducing
simulations and bisimulations for alternating-time temporal logic
under imperfect information. We prove that these (bi)simulations
preserve the interpretation of formulas for the whole syntax of
ATL$^*$, when interpreted with imperfect information and imperfect
recall, for both the objective and subjective variants of the
semantics \cite{BullingJamroga14,Jamroga15specificationMAS}. Most
interestingly for MAS verification, we apply the (bi)simulations to
the model reduction of a class of
electronic voting protocols without encryption.

Electronic voting is often considered as an attractive alternative to paper-based elections and referendums due to a number of advantages:
increased accessibility, availability, voter turnout, cost-efficiency, and usability, as well as speed and accuracy of the voting, counting and publication processes.
However, electronic voting poses a number of challenges: resistance to
coercion and other types of fraud, secrecy, anonymity, verifiability,
democracy (the right to vote at most once), accountability, etc.
Those threats exist also in paper voting, but the use of technology
magnifies their scope and potential impact.  There are also other
issues, specific to electronic voting, such as limited access to
technology and complex interaction between technology and public
understanding and trust in the
procedure~\cite{computer-ate-my-vote,Ryan15verifiability}.  In this
paper, we focus on the property of \emph{coercion resistance} that
captures the voters' ability to choose freely how they vote: whatever
course of action the coercer adopts, the voter always has a strategy
to vote as they intend while appearing to comply with the coercer's
requirements.

An increasing amount of research has focused recently on the
verification of many of these properties for various types of voting
protocols \cite{Beckertetal2014,Cortier2015}.
The frameworks used for modeling and verifying security properties of
voting protocols include, to mention only a few, process calculi such
as the {\em applied $\pi$-calculus} or \emph{CSP}
\cite{DKR09,Hoare78,SS96}, rewriting-based approaches
\cite{CDLMS99,DM02,BJL12}, approaches based on flat transition systems
etc.

In this paper, we show how our bisimulation can be used to obtain significant model reductions for some voting protocols.
In consequence, we develop a verification procedure for those voting protocols, based on the multi-agent approach.
The main advantage of multi-agent logics is the provision of a unified specification language for a vast array of properties.
Due to that feature, the logics can for instance serve to disambiguate the variety of informal intuitions behind coercion resistance that are found in the
literature, cf.~\cite{Tabatabaei16expressing}.
Moreover, multi-agent logics allow for a more flexible and expressive specification of variants of coercion-resistance, involving explicit references to the dynamics of attacker's knowledge, as well as the (non)existence of strategies suitable for the attacker and/or the voter.
The same applies to specification of other important properties, such as individual verifiability, end-to-end (universal) verifiability, and accountability.

Here, we focus on a particular formalization of coercion resistance, and verify it for a simplified version of the ThreeBallot voting protocol \cite{Rivest06,Rivest07threeProtocols}.
Our first results in this line
already allow for verification of models with a larger number of
voters and candidates compared to the approach based on process
calculi in~\cite{Moran2014,Moran2016}.  Even more importantly, our
experiments show that bisimulation-based reductions can turn
verification from a utterly difficult task to something feasible in
practice.

\subsection{Structure of the Paper and Technical Contributions}

The article is structured as follows.  In Section~\ref{setting}, we
introduce the syntax and semantics of ATL$^*$ interpreted under the
assumption of imperfect information and imperfect recall. In
Section~\ref{bisimulation}, we propose the novel relations of
simulation and bisimulation for the setting. Moreover, we prove the
main theoretical result of this paper, namely that the bisimulation
preserves the interpretation of ATL$^*$ formulas. In
Section~\ref{sec:HM}, we show that our notion of bisimulation does not
enjoy the Hennessy-Milner property, and we discuss some necessary
conditions for two models to be logically equivalent. The results
highlight some interesting features of the current semantics of
ATL$^*$ with imperfect information, and might point to a novel
semantics for logics of strategic ability. Further, in
Section~\ref{3ballot}, we present our case study based on the
ThreeBallot voting protocol. We formalize ThreeBallot as a game
structure, and provide two abstractions of the protocol. Then, we
point out that the abstractions are indeed bisimilar with the original
game structure. Finally, in Section~\ref{exp}, we evaluate the gains
in terms of verification time, by comparing model checking of the
abstractions to model checking of the original model. We conclude in
Section~\ref{conc}, and indicate some future directions of research.

\para{Previous version of the material.}
This work extends and revises the results presented in the conference
paper~\cite{BelardinelliCDJJ17}. The main technical differences are as
follows.  First, our results are now given for the whole language of
ATL$^*$, while \cite{BelardinelliCDJJ17}~considered only its
restricted fragment ATL. Secondly, we correct a mistake in the main
preservation result. In the previous version, we erroneously used a
weaker notion of bisimulation, which actually does not preserve even
formulas of ATL. We discuss the details, and point out some
differences between the two definitions in
Section~\ref{sec:comparison}.  Thirdly, we report completely new
theoretical research in Section~\ref{sec:complexity} (computational
complexity of checking for bisimulation) and Section~\ref{sec:HM}
(Hennessy-Milner property and necessary conditions for logical
equivalence of models).  Finally, we revise the implementation of our
ThreeBallot models, and present entirely new experimental results.  In
particular, we extend our model checking experiments to the imperfect
information semantics of ATL, for which the bisimulation-based
reductions were in fact designed.

Regarding the presentation, we add a running example, and streamline
the introduction of the bisimulation, in order to make the concept
easier to understand.

\subsection{Related Work}

The literature on both logics for strategies and the formal
verification of voting protocols is extensive and rapidly
growing. Hereafter we only consider the works most closely related to
the present contribution.

{\bf Bisimulations for ATL.} An in-depth study of model equivalences
induced by various temporal logics appears
in \cite{Goltz92equivalences}. Bisimulations for ATL$^*$ with perfect
information have been introduced in \cite{AlurHKV98}. Since then there
have been various attempts to extend these to more expressive
languages (including Strategy Logic recently \cite{BelardinelliDM18}),
as well as to contexts of imperfect
information \cite{Agotnes07irrevocable,Dastani10programs-aamas}.
In \cite{Dastani10programs-aamas,Melissen13phd} non-local model
equivalences for ATL with imperfect information have been put
forward. However,
these works do not deal with the
imperfect information/imperfect recall setting here considered, nor do
they provide a local account of bisimulations.  Finally,
in \cite{BelardinelliCDJJ17} we provided a different notion of
bisimulation for ATL only with imperfect information and imperfect
recall. We remarked in the introduction that such definition was
flawed and provided a counterexample in
Section~\ref{sec:comparison}. The present contribution is also aimed at
rectifying this result.

{\bf Verification of Voting Protocols.} Our work is inspired by recent
contributions on the verification of voting protocols, mostly by using
the $\pi$-calculus and \emph{CSP}~\cite{DKR09,Hoare78,SS96}.
Some existing attempts at verification of e-voting protocols follow this approach by using one of the security verifiers~\cite{Ciobaca11,MengHW11}, or even a general equivalence checker for a process algebra~\cite{Moran2014,Moran2016,DelauneKR09}, where privacy-type and anonymity properties of the protocols are verified by using CSP.
An important strand in verification of voting protocols is based on specifications in first-order logic or linear logic. The verification is done by means of theorem proving~\cite{DeYoungS11,BeckertGS13b,PattinsonS15,DrewsenDS17} or bounded model checking (BMC)~\cite{Beckertetal2014,BeckertGS13a,BeckertGSBW14,Schurmann16}. In particular~\cite{Schurmann16} presents a BMC-based analysis of risk-limiting audits and~\cite{DrewsenDS17}
applies theorem proving to automated verification of the Selene voting protocol.
In \cite{Beckertetal2014} the authors define two semantic criteria for
single transferable vote (STV) schemes,
then show how BMC and SMT solvers can be used to check whether these
criteria are met.
In
\cite{Moran2016} the authors construct CSP models of the
\ThreeBallot system and use them to produce an automated formal
analysis of their anonymity properties.  
Finally, some challenges and solutions for verification of end-to-end verifiable systems were addressed in~\cite{Schurmann13,Cortier15}.

An issue that can be identified with many of the above approaches is
that the description of the system, and the property to be verified,
are not clearly distinguished.  We emphasize that multi-agent logics
allow for a clear separation of the two.  Moreover, they provide a
wider variety of properties, including ones that are related to the
existence of the attacker's strategies.  Last but not least, they give
reasonable hope for interesting verification performance. In our
experiments, we have been able to model-check \ThreeBallot models with
5 voters and 2 candidates, or 4 candidates and 3 voters, while
e.g.~the results in \cite{Moran2016} are provided for at most 3 voters
and 2 candidates.  In this respect, the closest work to ours is the
recent paper~\cite{Jamroga18Selene}, presenting some experimental
results for verification of the voting protocol SELENE, based on the
multi-agent logic ATL$_{ir}$.

\section{The Formal Setting} \label{setting}

In this section, we introduce the syntax of the Alternating-time
Temporal Logic ATL$^*$ \cite{AlurHenzingerKupferman02}, and present
its semantics for agents with imperfect information and imperfect
recall, defined on imperfect information concurrent game structures
(iCGS).  The assumption of \emph{imperfect information} means that
agents can only partially observe the global state of the
system. Thus, they have to make their decisions, and execute their
strategies, with only partial knowledge about the current
situation. The assumption of \emph{imperfect recall} means that they
do not necessarily memorize all of their past observations. This does
not mean that the agents have no memory at all. However, an agent's
recall of the past must be encapsulated in the current local state of
the agent.

The formal definitions and notation follow~\cite{DimaT11}.

\subsection{Concurrent Game Structures}

Concurrent game structures have originally been introduced
in \cite{AlurHenzingerKupferman02} in a perfect information
setting. Here we consider their version for contexts of imperfect
information~\cite{Jamroga+04b}.

\begin{definition} \label{iCGS}
A {\em concurrent game structure with imperfect information}, or iCGS,
is a tuple $\G = \langle Ag, AP, S, s_0, Act,  \{\sim_i \}_{i \in Ag},
d,\to,  \pi \rangle$ such that
\begin{itemize}
\item $Ag$ is a nonempty and finite set of \emph{agents}. Subsets $A \subseteq Ag$ of agents are called \emph{coalitions}.
\item $AP$ is a set of \emph{atomic propositions}, or atoms.
\item  $S$ is a non-empty set of \emph{states} and $s_0 \in S$ is the \emph{initial state} of $\G$.
\item $Act$ is a finite non-empty set of {\em actions}. A tuple $\vec a  = ( a_i)_{i \in Ag} \in Act^{Ag}$
is called a \emph{joint action}.

\item For every agent $i \in Ag$, $\sim_i$ is an equivalence relation on $S$,
called the \emph{indistinguishability relation} for $i$.

\item $d: Ag \times S \to (2^{Act} \setminus \{ \emptyset \})$ is the \emph{protocol function},
satisfying the property that, for all states $s, s'\in S$ and any agent $i$, $s \sim_i s'$ implies $d(i, s) = d(i,
s')$. That is, the same (non-empty) set of  actions is available to agent $i$ in indistinguishable states.
\item $\to \subseteq S \times Act^{Ag} \times S$ is the
  {\em transition relation} such that, for every state $s \in S$ and
	joint action $\vec{a} \in Act^{Ag}$, $(s,\vec{a},s') \in \to$
 for some state $s' \in S$ iff $a_i \in d(i, s)$ for every agent $i \in
	Ag$.  We normally write $s \xrightarrow{\vec{a}} r$ for
	$(s,\vec{a},r) \in \to$.
\item
$\pi: S \to 2^{AP}$ is the {\em state-labeling function}.
\end{itemize}
\end{definition}

By Def.~\ref{iCGS} at any given state $s$, every agent $i \in Ag$ can
perform the enabled actions in $d(i,s)$. A joint action $\vec{a}$
fires a transition from state $s$ to some state $s'$ only if each
$a_i$ is enabled for agent $i$ in $s$. Further, every agent $i$ is
equipped with an indistinguishability relation $\sim_i$, with
$s \sim_i s'$ meaning that $i$ cannot tell state $s$ from state $s'$,
i.e., agent $i$ possesses the same information, makes the same
observation in the two states. In particular, the same actions are
enabled in indistinguishable states.

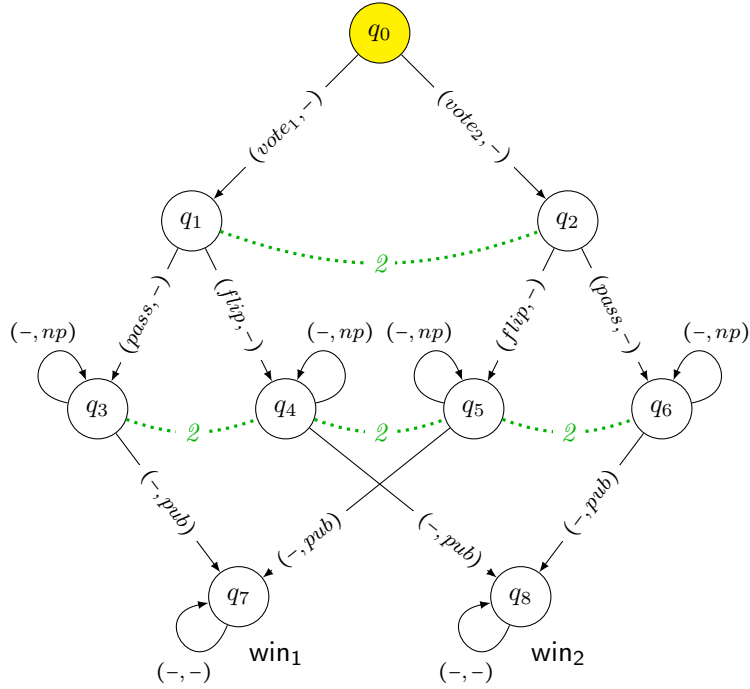
\begin{figure}[t]\centering
\begin{tikzpicture}[>=latex,scale=1.25]
 
    \tikzstyle{state}=[circle,draw,trans, minimum size=8mm]
    \tikzstyle{initstate}=[circle,draw,trans, minimum size=8mm, fill=yellow]
    \tikzstyle{trans}=[font=\footnotesize]
    \tikzstyle{epist}=[-,dotted,darkgreen,very thick,font=\footnotesize]
    \tikzstyle{loop}=[-latex]

    \path (0,0) node[initstate] (qinit) {$q_{0}$}
          (-2,-2) node[state]  (qv1) {$q_{1}$}
          (2,-2) node[state]  (qv2) {$q_{2}$}
          (-3,-4) node[state]  (qv1pass) {$q_{3}$}
          (-1,-4) node[state]  (qv1flip) {$q_{4}$}
          (1,-4) node[state]  (qv2flip) {$q_{5}$}
          (3,-4) node[state]  (qv2pass) {$q_{6}$}
          (-1.5,-6) node[state]  (qv1pub) {$q_{7}$}
									  +(0.4,-0.6) node {\small $\prop{win_1}$}
          (1.5,-6) node[state]  (qv2pub) {$q_{8}$}
									  +(0.4,-0.6) node {\small $\prop{win_2}$}
        ;

    \path[->,font=\scriptsize] (qinit)
          edge
            node[midway,sloped]{\onlabel{$(vote_1,-)$}} (qv1)
          edge
            node[midway,sloped]{\onlabel{$(vote_2,-)$}} (qv2);

    \path[->,font=\scriptsize] (qv1)
          edge
            node[midway,sloped]{\onlabel{$(flip,-)$}} (qv1flip)
          edge
            node[midway,sloped]{\onlabel{$(pass,-)$}} (qv1pass);
    \path[->,font=\scriptsize] (qv2)
          edge
            node[midway,sloped]{\onlabel{$(pass,-)$}} (qv2pass)
          edge
            node[midway,sloped]{\onlabel{$(flip,-)$}} (qv2flip);

    \path[->,font=\scriptsize] (qv1pass)
          edge
            node[midway,sloped]{\onlabel{$(-,pub)$}} (qv1pub);
    \path[->,font=\scriptsize] (qv1flip)
          edge
            node[near end,sloped]{\onlabel{$(-,pub)$}} (qv2pub);
    \path[->,font=\scriptsize] (qv2pass)
          edge
            node[midway,sloped]{\onlabel{$(-,pub)$}} (qv2pub);
    \path[->,font=\scriptsize] (qv2flip)
          edge
            node[near end,sloped]{\onlabel{$(-,pub)$}} (qv1pub);

    \path[epist] (qv1)
          edge[bend right=\bend]
            node[midway,sloped]{\onlabel{\color{darkgreen} 2}} (qv2);
    \path[epist] (qv1pass)
          edge[bend right=\bend]
            node[midway,sloped]{\onlabel{\color{darkgreen} 2}} (qv1flip);
    \path[epist] (qv1flip)
          edge[bend right=\bend]
            node[midway,sloped]{\onlabel{\color{darkgreen} 2}} (qv2flip);
    \path[epist] (qv2flip)
          edge[bend right=\bend]
            node[midway,sloped]{\onlabel{\color{darkgreen} 2}} (qv2pass);

    \draw[loop](qv1pass) ..controls +(-1,0.2) and +(-0.4,1).. (qv1pass) node[midway,above=3pt,font=\scriptsize] {{$(-,np)$}};
    \draw[loop](qv1flip) ..controls +(1,0.2) and +(0.4,1).. (qv1flip) node[midway,above=3pt,font=\scriptsize] {{$(-,np)$}};
    \draw[loop](qv2pass) ..controls +(1,0.2) and +(0.4,1).. (qv2pass) node[midway,above=3pt,font=\scriptsize] {{$(-,np)$}};
    \draw[loop](qv2flip) ..controls +(-1,0.2) and +(-0.4,1).. (qv2flip) node[midway,above=3pt,font=\scriptsize] {{$(-,np)$}};
    \draw[loop](qv1pub) ..controls +(-0.4,-1) and +(-1,-0.2).. (qv1pub) node[midway,below=3pt,font=\scriptsize] {{$(-,-)$}};
    \draw[loop](qv2pub) ..controls +(-0.4,-1) and +(-1,-0.2).. (qv2pub) node[midway,below=3pt,font=\scriptsize] {{$(-,-)$}};
 \end{tikzpicture}
\caption{Simple model of 2-stage voting. Dotted lines represent indistinguishability for agent 2 (the election official).}
\label{fig:voting-model}
\end{figure}

\begin{example}\label{ex:voting-model}
Consider a two-stage voting system with a simple anti-coercion mechanism. In the first stage, the voter casts her vote for one of the candidates; we assume for simplicity that there are only two candidates in the election. In the second stage, the voter has the option to send a signal -- by an independent communication channel -- that flips her vote to the other candidate.
After that, the election official publishes the result of the election (action \emph{pub}). However, a dishonest official can also decide not to publish the outcome (action \emph{np}).
An iCGS modeling the scenario for a single voter (represented by agent $1$) is depicted in Figure~\ref{fig:voting-model}. The election official is represented by agent $2$; the dotted lines indicate that the official does not directly see the choices made by the voter.
Note that this toy model is in fact a turn-based game, i.e., in every state only a single agent has a choice of decision while the other agent is passive.
\end{example}

\para{Runs.}
Given an iCGS $\G$ as above, a {\em run} is a finite or infinite
sequence $\lambda = s_0 \vec{a}_0 s_1 \ldots$ in $((S \cdot
Act^{Ag})^*\cdot S) \cup (S \cdot Act^{Ag})^\omega$ such that for
every $j\geq 0$, $s_j \xrightarrow{\vec{a}_j} s_{j+1}$. Given a run
$\lambda = s_0 \vec{a}_0 s_1 \ldots$ and $j \geq 0$, $\lambda[j]$
denotes the $j+1$-th state $s_j$ in the sequence; while $\lambda_{\geq
j}$ denotes run $s_j \vec{a}_j s_{j+1} \ldots$ starting from
$\lambda[j]$.  We denote by $Run(\G)$ the set of all runs in iCGS
$\G$.
For a coalition $A \subseteq Ag$ of agents, a {\em joint $A$-action}
denotes a tuple $\vec a_A = (a_i)_{i \in A} \in Act^A$ of actions, one
for each agent in $A$.  For coalitions $A\subseteq B \subseteq Ag$ of
agents, a joint $A$-action $\vec a_A$ \emph{is extended} by a joint
$B$-action $\vec b_B$, denoted $\vec a_A \sqsubseteq \vec b_B$, if for
every $i \in A$, $a_i = b_i$.  Also, a joint $A$-action $\vec
a_A$ is \emph{enabled} at state $s\in S$ if for every agent $i \in A$,
$a_i \in d(i,s)$.

\para{Epistemic neighbourhoods.}
Given a coalition $A \subseteq Ag$ of agents,
the \emph{collective
knowledge relation} $\sim^E_A$ is defined as $\bigcup_{i \in
A} \sim_i$, while the \emph{common knowledge relation} $\sim^C_A$ is
the transitive closure $(\bigcup_{i \in A} \sim_i)^+$ of $\sim^E_A$.
Then,
$C_A^{\G}(q) = \{ q' \in S \mid q' \sim_A^C q \}$
is the {\em common knowledge neighbourhoods} (CKN),
of state $q$ for coalition $A$ in the iCGS $\G$.  We will omit the
superscript $\G$ whenever it is clear from the context.

\para{Uniform strategies.}
We now recall a notion of strategy adapted to iCGS with imperfect
information \cite{Jamroga+04b}.
\begin{definition} \label{unstrategy}
A \emph{(uniform)} {\em strategy} for an agent $i \in Ag$ is a
function $\sigma : S \to Act$ that is compatible with $d$ and
$\sim_i$, i.e.,
\begin{itemize}
\item for every state $s \in S$, $\sigma(s) \in d(i, s)$;
\item for all states $s, s' \in S$, $s \sim_i s'$ implies $\sigma(s) = \sigma(s')$.
\end{itemize}
\end{definition}

By Def.~\ref{unstrategy} a strategy has to be uniform in the sense
that in indistinguishable states it must return the same action. Such
strategies are also known as {\em observational} in the literature on
game theory and control theory.  Note that in this paper we use
memoryless strategies, whereby only the current state determines the
action to perform.
This choice is dictated by the application
in hand, namely voting protocols,
where the unbounded recall is not needed.
Namely, the information available to the agents is fully encoded
in global states together with help of indistinguishability relations\footnote{Therefore memoryless strategies
already encode the agent's memory of all her past observations.}.
Note that the finite memory of a given depth of recall can be obtained in
practical setting by introducing an adequate number of special variables,
i.e., dimensions in iCGSs whose states are valuations of variables.

Perfect recall strategies with imperfect information can be defined
similarly, as memoryless strategies on tree unfoldings of iCGS. We
leave this extension for future work.

A strategy for a coalition $A$ of agents is a set $\sigma_A
= \{ \sigma_a \mid a \in A\}$ of strategies, one for each agent in
$A$.
Given coalitions $A \subseteq B \subseteq Ag$, a strategy $\sigma_A$
for coalition $A$, a state $s \in S$, and a joint $B$-action $\vec b_B \in
Act^B$ that is enabled at $s$, we say that $\vec b_B$
is \emph{compatible with} $\sigma_A$ ({\em in} $s$)
whenever $\sigma_A(s) \sqsubseteq \vec
b_B$.
For states $s, s' \in S$ and strategy $\sigma_A$, we write
$s \xrightarrow{\sigma_A(s)} r$ if $s \xrightarrow{\vec a} r$ for
some joint action $\vec a \in Act^{Ag}$ that is compatible with
$\sigma_A$.

There exist two alternative semantics of strategic operators under imperfect information, corresponding to two different notions of success for a strategy. The idea of \emph{subjective} ability~\cite{Schobbens04ATL,Jamroga+04b} requires that a winning strategy must succeed from all the states that the coalition considers possible in the initial state of the play. The alternative notion of \emph{objective} ability~\cite{BullingJamroga14} assumes that it suffices for the strategy to succeed from the initial state alone.
Accordingly, we define two notions of the \emph{outcome} of a strategy $\sigma_A$ at
state $s$, corresponding to the
{objective} and the {subjective} interpretation of alternating-time operators.
Fix a state $s$ and a strategy $\sigma_A$ for coalition $A$.
\begin{enumerate}
\item The set
of \emph{objective outcomes of $\sigma_A$ at $s$}
is defined as
$out^\G_{obj} (s, \sigma_A) = \big \{
\lambda \in Run(\G) \mid \lambda[0] = s \text{ and for all } j \geq 0,
\lambda[j] \xrightarrow{\sigma_A(\lambda[j])} \lambda[j + 1]\big\}
$.

\item The set
of \emph{subjective outcomes of $\sigma_A$ at $s$} is defined as
$
out^\G_{subj} (s, \sigma_A) =
\displaystyle\bigcup_{i \in A, s' \sim_i s} out^\G_{obj}(s', \sigma_A)$.
\end{enumerate}

\subsection{Alternating-Time Temporal Logic.}

We now introduce the alternating-time temporal logic ATL$^*$ to be interpreted on iCGS.
\begin{definition}[ATL$^*$]
The language of ATL$^*$ is formally defined by the following grammar, where $p \in AP$ and $A \subseteq Ag$:
\begin{eqnarray*}
\phi  & ::= &
p \mid \neg \phi \mid \phi \to \phi \mid
  \llangle A \rrangle \psi\\
\psi  & ::= &
\phi \mid \neg \psi \mid \psi \to \psi \mid
   X \psi \mid \psi U \psi
\end{eqnarray*}
Formulas $\phi$ are called \emph{state formulas of ATL$^*$}, or simply \emph{formulas of ATL$^*$}.
Formulas $\psi$ are sometimes called \emph{path formulas of ATL$^*$}. 
\end{definition}

The ATL$^*$ operator $\llangle A \rrangle$ intuitively means that `the
agents in coalition $A$ have a (collective) strategy to achieve \ldots',
where the goals are LTL formulas built by using operators `next' $X$ and
`until' $U$.
We
define $A$-formulas as the formulas in ATL$^*$ in which $A$ is the only
coalition appearing in ATL$^*$ modalities.

Traditionally, ATL$^*$ under imperfect information has been given either
state-based or history-based semantics, and several variations have
been considered on the interpretation of strategy operators. Here we
present both the  objective and subjective
variants of the state-based semantics with imperfect information and
imperfect recall.
\begin{definition} \label{semantics}
Given an iCGS $\G$, a state formula $\phi$, and path formula $\psi$, the \emph{subjective}
(resp. \emph{objective}) satisfaction of $\phi$ at state $s$ and of $\psi$ in path $\lambda$, denoted
$(\G, s) \!\models_{x} \!\phi$ and $(\G, \lambda) \!\models_{x} \!\psi$ for $x \!= \!subj$ (resp.~$x \!= \!obj$), is
defined recursively as follows:
\begin{align*}
& (\G, s)  \models_{x}  p & \text{ iff } &  p \in \pi(s) \\
& (\G, s)  \models_{x}  \lnot \phi & \text{ iff } & (\G, s) \not \models_x \phi\\
&(\G, s)  \models_{x} \phi \to \phi' & \text{ iff } & (\G, s) \not \models_x \phi \text{ or } (\G, s) \models_x \phi'\\
&(\G, s) \models_{x} \llangle A \rrangle \psi & \text{ iff } & \text{ for some } \sigma_{A}, \text{ for all } \lambda \in out^\G_{x}(s, \sigma_A),
(\G, \lambda) \models_{x} \psi\\
&(\G, \lambda) \models_{x} \phi & \text{ iff } & (\G, \lambda[0]) \models_{x} \phi\\
& (\G, \lambda)  \models_{x}  \lnot \psi & \text{ iff } & (\G, \lambda) \not \models_x \psi\\
&(\G, \lambda)  \models_{x} \psi \to \psi' & \text{ iff } & (\G, \lambda) \not \models_x \psi \text{ or } (\G, \lambda) \models_x \psi'\\
&(\G, \lambda) \models_{x}  X \psi & \text{ iff } &  (\G, \lambda_{\geq 1}) \models_{x}   \psi\\
&(\G, \lambda) \models_{x}  \psi U \psi' & \text{ iff } & \text{ for some } j \geq 0, (\G, \lambda_{\geq j}) \models_{x} \psi'\\
 & & & \text{ and for all } k,  0 \leq k <j \text{ implies } (\G, \lambda_{\geq k}) \models_{x} \psi
\end{align*}
\end{definition}

\begin{example}\label{ex:voting-atl}
Consider the simple voting model in Example~\ref{ex:voting-model}.
Clearly, the voter cannot enforce a win of any candidate, since the election official might block the publication of the outcome:
$\G,q_0 \models \neg\coop{1}\Sometm\prop{win_1} \land \neg\coop{1}\Sometm\prop{win_2}$.
On the other hand, she can prevent any given candidate from being elected, by casting her vote for the other candidate:
$\G,q_0 \models \coop{1}\Always\neg\prop{win_1} \land \coop{1}\Always\neg\prop{win_2}$.
Finally, the voter and the official together can get any arbitrary candidate elected:
$\G,q_0 \models \coop{1,2}\Sometm\prop{win_1} \land \coop{1,2}\Sometm\prop{win_2}$.
Note that the truth value of formulae in $q_0$ does not depend on whether we use the subjective or the objective semantics.

The difference between the two semantics can be demonstrated, e.g., in state $q_3$. In that state, the election official has the objective ability to make candidate 1 win
($\G,q_3 \models_{obj} \coop{2}\Sometm\prop{win_1}$) by playing \emph{pub} regardless of anything.
On the other hand, there is no uniform strategy for $2$ that guarantees the same from all the states indistinguishable from $q_3$ (i.e., from $q_3, q_4, q_5, q_6$).
In consequence, $\G,q_3 \models_{subj} \neg\coop{2}\Sometm\prop{win_1}$.
\end{example}

The individual knowledge operator $K_i$ can be added to the syntax
of ATL$^*$ with the following semantics:
\[
 (\G, s) \models_{x} K_i \phi \text{ iff for all } s'\in S,
 s' \sim_i s \text{ implies } (\G, s') \models_{x} \phi
\]

By considering the subjective interpretation of ATL$^*$, this operator can
be derived: $(\G,s) \models_{subj} K_i \phi $ iff
$(\G,s) \models_{subj} \llangle i \rrangle \phi U \phi$.
There exists no such definition for the knowledge operator in ATL$^*$ with the objective semantics.
The latter is easy to see, as the objective semantics of $\coop{i}$ in $(\G,s)$ does not refer in any way to the states epistemically indistinguishable from $s$.

\section{Simulations and Bisimulations} \label{bisimulation}

In this section we introduce a notion of bisimulation for imperfect
information concurrent game structures.  Then, we show that it
preserves the meaning of formulas in ATL$^*$, when interpreted under
the assumptions of imperfect information and imperfect recall,
introduced in Section~\ref{setting}.

\subsection{Bisimulation for ATL$^*$ with Imperfect Information and Imperfect Recall}

We start with some auxiliary notions and definitions.

\para{Partial strategies and outcome states.}
A \emph{partial (uniform) strategy} for agent $i \in Ag$ is a partial
function $\sigma : S \rightarrow Act$ such that for each $s, s' \in
S$, if $s \sim_i s'$ then $\sigma(s) = \sigma(s')$.  We denote
the domain of the partial strategy $\sigma$ as $dom(\sigma)$.  Given a
coalition $A \subseteq Ag$, a \emph{partial strategy for
$A$} is a tuple $(\sigma_i)_{i \in A}$ of partial
strategies, one for each agent $i \in A$.
The set of partial uniform strategies for $A$ is denoted $PStr_A$.
Given set $Q\subseteq S$ of states and coalition $A\subseteq Ag$,
we denote by $PStr_A(Q)$ the set of partial uniform strategies whose domain is
$Q$:
\[
PStr_A(Q) = \big\{(\sigma_i)_{i \in A} \in PStr_A \mid dom(\sigma_i) =
Q \text{ for all } i \in A\big\}
\]
Additionally, given a (total or partial) strategy $\sigma_A$  and a
state $q\in dom(\sigma_A)$, define the set of \emph{successor
states of $q$ by $\sigma$} as $succ(q, \sigma_A) = \{ s \in S \mid
q \xrightarrow{\sigma_A(q)} s \}$, and put 
$succ(\sigma_A) = succ(dom(\sigma_A),\sigma_A) = \bigcup_{s \in dom(\sigma_A)} succ(s,\sigma_A)$.

\para{Strategy simulators.}
Let $\G,\G'$ be two iCGS. A \emph{simulator of partial strategies for
coalition $A$} from $\G$ into $\G'$ is a set $ST$ of functions $ST_{Q,Q'}: PStr_A(Q) \rightarrow PStr_A(Q')$ for some subsets
$Q \subseteq S$ and $Q'\subseteq S'$.  Intuitively, every $ST_{Q,Q'}$ maps
each partial strategy $\sigma_A$ defined on set $Q$ in the iCGS $\G$
into a ``corresponding'' strategy $\sigma'_A$ defined on $Q'$ in
$\G'$.
Typically, we will map strategies between the common knowledge
neighborhoods of ``bisimilar'' states in $\G$ and $\G'$.  We formalize
this idea as follows.  Let $R \subseteq S\times S'$ be some
relation between states in $\G$ and $\G'$.  A \emph{simulator of
partial strategies for coalition $A$ with respect to relation $R$} is
a family $ST$ of functions $ST_{C_A(q),C_A'(q')} :
PStr_A(C_A(q)) \rightarrow PStr_A(C_A'(q'))$ such that $q\in S$,
$q'\in S'$, and $q R q'$.
Note that, since the functions are indexed by equivalence classes of states, the following additional property is automatically guaranteed:
\begin{itemize}
\item for every $r \in S$, $r' \in S'$, if $r \in C_A(q)$, $r' \in C_A'(q')$,
and $r R r'$, then $ST_{C_A(q),C_A'(q')} = ST_{C_A(r),C_A'(r')}$.
\end{itemize}

\para{Simulation and bisimulation.}
We can now present our notions of simulation and bisimulation on iCGS.
\begin{definition}[Simulation] \label{def:bisim}
Let $\G = \langle Ag, AP, S, s_0,  Act, \{\sim_i \}_{i \in Ag}, d, \to, \pi \rangle$
and $\G' = \langle Ag, AP, S', s'_0,  Act', \{\sim'_i \}_{i \in Ag}, d', \to', \pi'
\rangle$
be two iCGS defined on the same sets $Ag$ of agents and $AP$ of atoms.
Let $A\subseteq Ag$ be a coalition of agents.  A relation
$\altsim{A} \subseteq S \times S'$ is a \emph{simulation for $A$} iff
\begin{enumerate}
\item There exists a simulator $ST$ of partial strategies for $A$ w.r.t.~$\altsim{A}$,
such that $q \altsim{A} q'$ implies that:\label{it:original}
\begin{enumerate}
\item\label{it:agree} $\pi(q) = \pi'(q')$;

\item for every $i \in A$ , $r' \in S'$, if $q' \sim'_i r'$ then for some
  $r \in S$, $q \sim_i r$ and $r \altsim{A} r'$.

\item For every states $r \in C_A(q)$, $r' \in C_A'(q')$ such that $r \altsim{A} r'$,
  for every partial strategy $\sigma_A \in PStr_A(C_A(q))$,
  and every state $s' \in succ(r', ST(\sigma_A))$,
  there exists a state $s\in succ(r, \sigma_A)$ such that $s \altsim{A} s'$.
\end{enumerate}
\item If $q_1 \altsim{A} q'$ and $q_2 \altsim{A} q'$, then $C_A(q_1) = C_A(q_2)$.\label{it:inject}
\end{enumerate}

A relation $\altbisim{A}$ is a {\em bisimulation} iff both $\altbisim{A}$
and its converse $\altbisim{A}^{-1} = \{ (q',q) \mid q \altbisim{A} q' \}$ are simulations.
\end{definition}

By Def.~\ref{def:bisim} if state $q'$ {\em simulates} $q$, i.e.,
$q \altsim{A} q'$, then 1.(a) $q$ and $q'$ agree on the interpretation
of atoms; 1.(b) $q$ simulates the epistemic transitions from $q'$,
that is, information encoded in states is preserved by (bi)similarity;
and 1.(c) for every partial strategy $\sigma_A$, defined on the common
knowledge neighborhood $C_A(q)$, we are able to find some partial
strategy $ST(\sigma_A)$ (the same for all states in $C_A(q)$) such
that the transition relations $\xrightarrow{ST(\sigma_A)}$ and
$\xrightarrow{\sigma_A}$ commute with the simulation relation
$\altsim{A}$. Moreover, (2) the simulation relation is injective when
lifted to common knowledge neighborhoods.

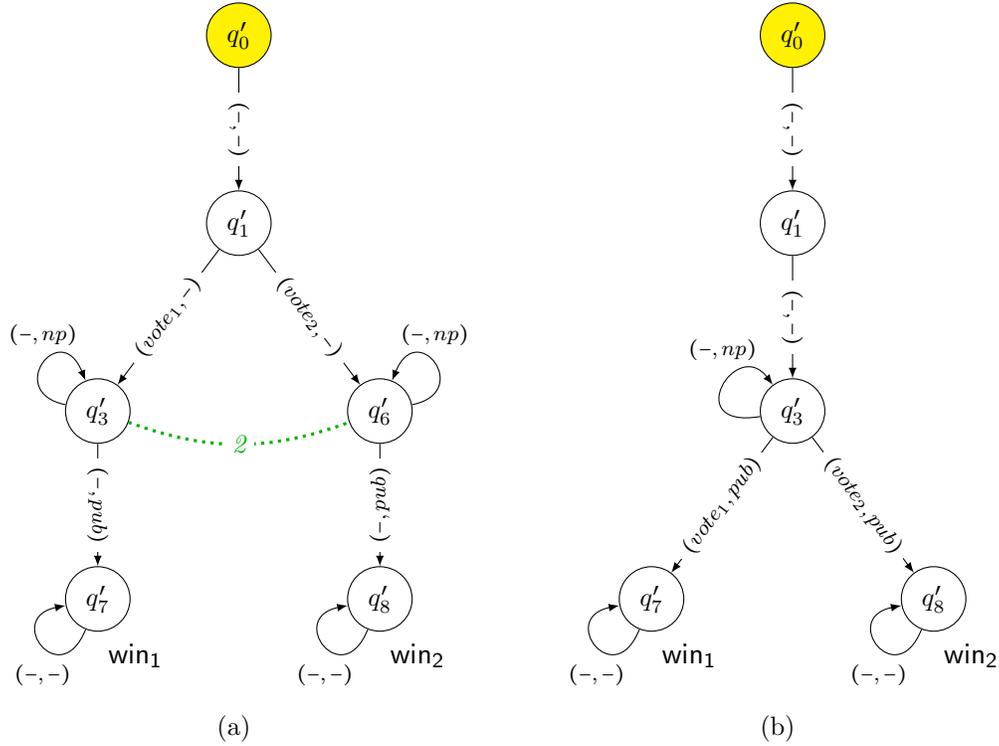
\begin{figure}[t]\centering
\subfigure[]{
\begin{tikzpicture}[>=latex,scale=1.25]
 
    \tikzstyle{state}=[circle,draw,trans, minimum size=8mm]
    \tikzstyle{initstate}=[circle,draw,trans, minimum size=8mm, fill=yellow]
    \tikzstyle{trans}=[font=\footnotesize]
    \tikzstyle{epist}=[-,dotted,darkgreen,very thick,font=\footnotesize]
    \tikzstyle{loop}=[-latex]

    \path (0,0) node[initstate] (qinit) {$q_{0}'$}
          (0,-2) node[state]  (qv1) {$q_{1}'$}
          (-1.5,-4) node[state]  (qv1pass) {$q_{3}'$}
          (1.5,-4) node[state]  (qv2pass) {$q_{6}'$}
          (-1.5,-6) node[state]  (qv1pub) {$q_{7}'$}
									  +(0.4,-0.6) node {\small $\prop{win_1}$}
          (1.5,-6) node[state]  (qv2pub) {$q_{8}'$}
									  +(0.4,-0.6) node {\small $\prop{win_2}$}
        ;

    \path[->,font=\scriptsize] (qinit)
          edge
            node[midway,sloped]{\onlabel{$(-,-)$}} (qv1);

    \path[->,font=\scriptsize] (qv1)
          edge
            node[midway,sloped]{\onlabel{$(vote_1,-)$}} (qv1pass)
          edge
            node[midway,sloped]{\onlabel{$(vote_2,-)$}} (qv2pass);

    \path[->,font=\scriptsize] (qv1pass)
          edge
            node[midway,sloped]{\onlabel{$(-,pub)$}} (qv1pub);
    \path[->,font=\scriptsize] (qv2pass)
          edge
            node[midway,sloped]{\onlabel{$(-,pub)$}} (qv2pub);

    \path[epist] (qv1pass)
          edge[bend right=\bend]
            node[midway,sloped]{\onlabel{\color{darkgreen} 2}} (qv2pass);

    \draw[loop](qv1pass) ..controls +(-1,0.2) and +(-0.4,1).. (qv1pass) node[midway,above=3pt,font=\scriptsize] {{$(-,np)$}};
    \draw[loop](qv2pass) ..controls +(1,0.2) and +(0.4,1).. (qv2pass) node[midway,above=3pt,font=\scriptsize] {{$(-,np)$}};
    \draw[loop](qv1pub) ..controls +(-0.4,-1) and +(-1,-0.2).. (qv1pub) node[midway,below=3pt,font=\scriptsize] {{$(-,-)$}};
    \draw[loop](qv2pub) ..controls +(-0.4,-1) and +(-1,-0.2).. (qv2pub) node[midway,below=3pt,font=\scriptsize] {{$(-,-)$}};
 \end{tikzpicture}
\label{fig:votingmodel-bisim1}
}
\hspace{.5cm}
\subfigure[]{
\begin{tikzpicture}[>=latex,scale=1.25]
 
    \tikzstyle{state}=[circle,draw,trans, minimum size=8mm]
    \tikzstyle{initstate}=[circle,draw,trans, minimum size=8mm, fill=yellow]
    \tikzstyle{trans}=[font=\footnotesize]
    \tikzstyle{epist}=[-,dotted,darkgreen,very thick,font=\footnotesize]
    \tikzstyle{loop}=[-latex]

    \path (0,0) node[initstate] (qinit) {$q_{0}'$}
          (0,-2) node[state]  (qv1) {$q_{1}'$}
          (0,-4) node[state]  (qv1flip) {$q_{3}'$}
          (-1.5,-6) node[state]  (qv1pub) {$q_{7}'$}
									  +(0.4,-0.6) node {\small $\prop{win_1}$}
          (1.5,-6) node[state]  (qv2pub) {$q_{8}'$}
									  +(0.4,-0.6) node {\small $\prop{win_2}$}
        ;

    \path[->,font=\scriptsize] (qinit)
          edge
            node[midway,sloped]{\onlabel{$(-,-)$}} (qv1);

    \path[->,font=\scriptsize] (qv1)
          edge
            node[midway,sloped]{\onlabel{$(-,-)$}} (qv1flip);

    \path[->,font=\scriptsize] (qv1flip)
          edge
            node[midway,sloped]{\onlabel{$(vote_1,pub)$}} (qv1pub)
          edge
            node[midway,sloped]{\onlabel{$(vote_2,pub)$}} (qv2pub);

    \draw[loop](qv1flip) ..controls +(-1.2,-0.2) and +(-0.6,0.8).. (qv1flip) node[midway,above=6pt,font=\scriptsize] {{$(-,np)$}};
    \draw[loop](qv1pub) ..controls +(-0.4,-1) and +(-1,-0.2).. (qv1pub) node[midway,below=3pt,font=\scriptsize] {{$(-,-)$}};
    \draw[loop](qv2pub) ..controls +(-0.4,-1) and +(-1,-0.2).. (qv2pub) node[midway,below=3pt,font=\scriptsize] {{$(-,-)$}};
 \end{tikzpicture}
\label{fig:votingmodel-bisim2}
}
\caption{(a) Even simpler model of 2-stage voting (from the point of view of the election official).
(b) Yet simpler model of 2-stage voting.}
\end{figure}

\begin{example}\label{ex:voting-bisim}
Let us go back to the simple two-stage voting from Example~\ref{ex:voting-model} and Figure~\ref{fig:voting-model}.
We observe that the model is bisimilar for $A=\{2\}$ to the one in Figure~\ref{fig:votingmodel-bisim1}.
The bisimulation connects $q_0$ with $q_0'$; $q_1$ and $q_2$ with $q_1'$; $q_3$ and $q_4$ with $q_3'$; $q_5$ and $q_6$ with $q_6'$; $q_7$ with $q_7'$; and $q_8$ with $q_8'$.
As we will see in Section~\ref{sec:preservation}, this implies that the abilities of the election official in both models must be exactly the same.

The iCGS can be reduced even further, and still retain the same abilities of the singleton coalition $\{2\}$, see Figure~\ref{fig:votingmodel-bisim2}.
We leave it to the interested reader to find the bisimulation for $\{2\}$ between the two models.
\end{example}

\begin{remark}
Technically, condition (2) in Def.~\ref{def:bisim} is required
as is shown in the end of this section by a counterexample
(cf. Subsection~\ref{sec:comparison}).  The extra property (2) corrects the statement and proof
of Theorem 9 in \cite{BelardinelliCDJJ17}.  Notice that, for every
bisimulation $\altbisim{A}$ defined on $\G = \G'$, condition (2) says
that two bisimilar states must lie in the same common knowledge
neighborhood, which is a natural constraint considering that two
bisimilar states should satisfy the same formulas in ATL$^*$.
\end{remark}

\subsection{Preservation Theorem}\label{sec:preservation}

In order to show that bisimilar states satisfy the same formulas in
ATL$^*$, we prove the following auxiliary result.  Hereafter, runs
$\lambda \in S^+$, $\lambda' \in S'^+$ are $A$-{\em bisimilar}, or
$\lambda \altbisim{A} \lambda'$ iff for every $i \geq 0$,
$\lambda[i] \altbisim{A} \lambda'[i]$
\begin{proposition}\label{prop:infinite-strats}
If $q \altsim{A} q'$ then
for every strategy $\sigma_A$, there exists a
strategy $\sigma_A'$ such that
\begin{itemize}
\item[($*$)] for every run $\lambda' \in  out^{\G'}_x(q',\sigma_A')$, for $x \in \{subj,obj\}$,
there exists an infinite run $\lambda \in out^\G_x(q,\sigma_A) $ such
that $\lambda \altsim{A} \lambda'$.
\end{itemize}
\end{proposition}
\begin{proof}
First of all we define the sequence
$\big(dom^n(\sigma_A)\big)_{n \in \mathbb{N}}$, of sets of states in
$\G$ such that $s \in dom^n(\sigma_A)$ iff $s$ can be reached in at
most $n$ steps from $C_A(q)$ by applying actions compatible with
strategy $\sigma_A$:
\begin{align*}
 dom^0(\sigma_A) & =  C_A(q) \\
 dom^{n+1}(\sigma_A) & =  dom^n(\sigma_A) \cup \bigcup \big\{C_A(r) \mid r \in \bigcup_{s \in dom^n(\sigma_A)} succ (s, \sigma_A)\big\}
\end{align*}

Also, we denote by $\sigma^n_A$ the partial strategy resulting from
restricting $\sigma_A$ to $dom^n(\sigma_A)$.

We now define inductively
a sequence $(\br \sigma^n_A)_{n \in \mathbb{N}}$ of partial strategies
 in $\G'$
such that
$dom(\br \sigma^n_A) \subseteq dom(\br\sigma^{n+1}_A)$ for every
$n \in \mathbb{N}$.  These partial strategies will be constructed
using strategy $\sigma_A$ and mapping $ST$ from point (1) in
Def.~\ref{def:bisim}.
The desired sequence of partial strategies
$\br \sigma^n_A$, for $n \geq 1$, is defined as follows:
\begin{enumerate}
\item
$dom(\br \sigma_A^0) = C'_A(q')$ and $dom(\br \sigma_A^{n+1}) =
dom(\br \sigma_A^n) \cup \{ r' \mid C'_A(r') \cap
succ(\br \sigma_A^n) \neq \emptyset \}$;
\item for all $r' \in dom (\br \sigma_A^0)$,
$\br \sigma_A^0(r') = ST_{C_A\!(q),C_A'\!(q')}(\sigma^0_A)(r')$;
\item for all $r' \in dom(\br \sigma_A^{n+1})$,
\[ \hspace*{-30pt}
\br \sigma^{n+1}_A(r') =
\begin{cases}
\br \sigma^n_A(r') & \text{for } r'\in dom(\br \sigma^n_A) \\
ST_{C_A\!(r),C_A'\!(r')}(\sigma^{n+1}_A)(r') & \text{for }
 r'\not \in dom(\br\sigma^n_A), C_A'\!(r') \cap succ (\!\br\sigma^n_A\!) \neq \emptyset , \\
&
\text{and } r \text{ is any state in $\G$ s.t.~} r \altsim{A} r'
\end{cases}
\]
\end{enumerate}

We give first a number of properties for this sequence of partial strategies.
First note that, in the last line of the definition of $\br\sigma_A^{n+1}(r')$, $r$ can indeed be chosen arbitrarily,
since, by point (2) in Def.~\ref{def:bisim}, whenever $r_1 \altsim{A} r' $ and $r_2 \altsim{A} r'$ we must have
$C_A(r_1) = C_A(r_2)$, which implies that $ST_{C_A(r_1),C_A'(r')} =
ST_{C_A(r_2),C_A'(r')}$.

Further, if $C_A'(r') \cap dom (\br\sigma^n_A) \neq \emptyset$ then $
C_A'(r') \subseteq dom (\br\sigma^n_A) $.  This is trivial for
$n=0$. As for the induction step, if $C_A'(r') \cap dom (\br\sigma^{n+1}_A) \neq \emptyset$, but $C_A'(r') \not \subseteq dom (\br\sigma^{n}_A)$, then by induction hypothesis we have
$C_A'(r') \cap dom (\br\sigma^n_A) = \emptyset$. Further, by
definition of $\br \sigma^{n+1}_A$ we obtain $C_A'(r') \cap succ
(\br\sigma^n_A) \neq \emptyset$, which is the case for all
$r'' \in C_A'(r')$. Hence $C_A'(r') \subseteq
dom(\br \sigma^{n+1}_A)$.

We now show that, whenever we take some $u' \in C_A'(r') \cap succ
(\br\sigma^n_A) \neq \emptyset$ with $r'\not \in dom(\br\sigma^n_A)$,
we have $\{ u \in dom^{n+1}(\sigma_A) \mid u\altsim{A}
u'\} \neq \emptyset$ and, by the induction hypothesis
$\br \sigma^{n+1}_A$ is uniform over $dom(\br\sigma^n_A)$.  To see
the former, whenever $u'\in C_A'(\!r'\!) \cap succ (\!\br\sigma^n_A\!)$, one
can find $v'\!\in \!dom (\br \sigma^n_A)$ s.t.~$u' \in
succ(v',\br\sigma_A^n(v'))$.
By definition of $ \br \sigma_A^n$, there exists $v \altsim{A} v'$
s.t.~$v \in dom(\sigma^n_A)$ and $\br\sigma^n_A(v') =
ST_{C_A(v),C_A'(v')}(\sigma^n_A)(v')$.  From property 1.(c), this
implies the existence of $u$ with $u\altsim{A} u'$ and $u \in
succ(v, \sigma_A(v))$,
which entails $u \in dom(\sigma^{n+1}_A)$.

We then prove by induction on $n$ that $\br \sigma^n_A$ is uniform.
The case for $n = 0$ is immediate by application to $\sigma^0_A$ of
simulator $ST$ of partial strategies.  For $n > 0$, as noted above if
$C_A'(r') \cap dom (\br\sigma^n_A) \neq \emptyset$ for some $r'$, then
$C_A'(r') \subseteq dom (\br\sigma^n_A) $.  Therefore, the induction
step only needs proof for $r_1',r_2'\in
dom(\br \sigma^{n+1}_A) \setminus dom(\br \sigma^n_A)$.  So assume, in
this case, that $r_1' \sim_i r_2'$ for some $i \in A$, hence
$C_A'(r_1') = C_A'(r_2')$. Consider then some $r_1 \altsim{A} r_1'$
and $r_2 \altsim{A} r_2'$, which exist by the previous paragraph.
Then condition 1.(b) in Def.~\ref{def:bisim} implies the existence of
some $r_3$ such that $r_3 \sim_i r_1$ and $r_3 \altsim{A} r_2'$, which
then, by condition (2) implies $C_A(r_1) = C_A(r_3) = C_A(r_2)$, which
in turn
implies that $ST_{C_A(r_1),C_A'(r_1')} = ST_{C_A(r_2),C_A'(r_2')}$.
Then
$\br\sigma_A^{n+1}(r_1') = ST_{C_A(r_1),C_A'(r_1')}(\sigma_A^{n+1})(r_1') = ST_{C_A(r_2),C_A'(r_2')}(\sigma_A^{n+1})(r_2') =
\br \sigma_A^{n+1}(r_2')$ as $ST$ is a simulator of partial strategies which maps uniform strategies for $A$ in $\G$ to uniform
strategies for $A$ in $\G'$.

As a result, the ``limit'' partial strategy $\br \sigma_A =
{\displaystyle\bigcup_{n\in \mathbb{N}}} \br \sigma^n_A$ defined as
$\br\sigma_A(r') = \br \sigma^n_A(r')$ whenever $r' \in
dom(\br\sigma^n_A)$, is also uniform and has $dom(\br \sigma_A) =
succ(\br \sigma_A)$.  We then only need to transform it into a (total)
uniform strategy by imposing a fixed action $a_0 \in Act$ wherever
$\br \sigma^n_A$ is undefined, that is, introducing the following
uniform strategy $\sigma_A'$:
\[
\sigma_A'(r') = \begin{cases}
\br \sigma_A(r') & \text{ for } r' \in dom(\br \sigma_A) \\
a_0 & \text{ otherwise}
\end{cases}
\]

To prove property ($*$) for the subjective semantics, consider a run
$\lambda' \in out^{\G'}_{subj}(q',\sigma_A') $.
We build inductively the run $\lambda$ as follows: for the base
case, by condition 1.(b) in Def.~\ref{def:bisim} (and a short
induction on the length of the indistinguishability path connecting
$q'$ with $\lambda'[0]$), we obtain a state $r \in C_A(q)$ such that
$r \altsim{A} \lambda'[0]$. Then, set $\lambda[0] := r$.  For the
inductive step, assume $\lambda[k]$ has been built, with
$\lambda[j] \altsim{A} \lambda'[j]$ for all $j\leq k$.  Now, we have
$\lambda'[k+1] \in succ(\lambda'[k],\sigma'_A)$.  By definition,
$\sigma'_A(\lambda'[k]) =
ST_{C_A(\lambda[k]),C_A'(\lambda'[k])}(\sigma_A)(\lambda'[k])$ since
$\lambda[k] \altsim{A} \lambda'[k]$.  This, by property 1.(c) implies
the existence of some $u \in succ(\lambda[k],\sigma_A)$ such that
$u \altsim{A} \lambda'[k+1]$.  We then set $\lambda[k+1] := u$, and
obtain the desired result.  The same proof works for the objective
semantics, by simply starting the induction with $\lambda[0] = q$.
Property ($*$) is then proved for both semantics.
\end{proof}

By using Proposition~\ref{prop:infinite-strats} we are finally able to
prove the main preservation result of this paper.
\begin{theorem}\label{thm-bisim}
Let $\G$ and $\G'$ be iCGS, $q \in S$ and $q'\in S'$ be states such that
$q \altbisim{A} q'$, and $\lambda \in S^+$ and  $\lambda' \in S'^+$ be
runs such that $\lambda \altbisim{A} \lambda'$.  Then, for every
state $A$-formula $\phi$ and path $A$-formula $\psi$,
\begin{eqnarray*}
 (\G,q) \models \phi & \text{iff} & (\G',q') \models \phi\\
 (\G,\lambda) \models \psi & \text{iff} & (\G',\lambda') \models \psi
\end{eqnarray*}
\end{theorem}
\begin{proof}
The proof is by mutual induction on the structure of $\phi$ and
$\psi$.

The case for propositional atoms is immediate as $(\G, q) \models p$
iff $p \in \pi(q)$, iff $p \in \pi'(q')$ by definition of
bisimulation, iff $(\G',q') \models p$.  The inductive cases for
propositional connectives are also immediate.

For $\psi = \phi$, suppose that $(\G, \lambda) \models \psi$, that is,
$(\G, \lambda[0]) \models \phi$.  By assumption,
$\lambda[0] \altbisim{A} \lambda'[0]$ as well, and by induction
hypothesis $(\G', \lambda'[0]) \models \phi$. Thus,
$(\G', \lambda') \models \psi$.

For $\psi = X \psi'$, suppose that $(\G, \lambda) \models \psi$, that is,
$(\G, \lambda_{\geq 1}]) \models \psi'$.  By assumption, $\lambda_{\geq 1} \altbisim{A} \lambda'_{\geq 1}$ as well, and by induction hypothesis
$(\G', \lambda'_{\geq 1}) \models \psi'$. Thus,
$(\G', \lambda') \models \psi$.
The inductive cases for $\psi = \psi' U \psi''$ and $\psi = \psi' R \psi''$ are similar.

For $\phi = \llangle A \rrangle \psi$, $(\G,
q) \models \phi$ iff that for some strategy $\sigma_{A}$,
for all $\lambda \in out^{\G}_x(q, \sigma_A)$,
$(\G,\lambda) \models \psi$.
By Prop.~\ref{prop:infinite-strats}, there exists strategy $\sigma'_A$
s.t.~for all $\lambda' \in out^{\G'}_x(q', \sigma'_A)$, there exists
$\lambda \in out^{\G}_x(q, \sigma_A)$
s.t.~$\lambda \altbisim{A} \lambda'$.
By the induction hypothesis,
$(\G, \lambda) \models \psi$ iff $(\G', \lambda') \models \psi$.
Hence,
$(\G', q') \models \phi$.
\end{proof}

\begin{corollary}
Let $\G$ and $\G'$ be iCGS, and $q \in S$, $q'\in S'$ be states such
that $q \altbisim{A} q'$.  Then, for every $A$-formula $\varphi$,
\begin{eqnarray*}
 (\G, q) \models \varphi & \text{if and only if} & (\G', q') \models \varphi
\end{eqnarray*}
\end{corollary}

By Theorem~\ref{thm-bisim} we obtain that bisimilar states preserve
the interpretation of ATL$^\star$ formulas. More precisely, if states $q$ and
$q'$ are $A$-bisimilar then they satisfy the same $A$-formulas.
Observe that only $A$-formulas are normally preserved by
$A$-bisimulations. This is already a feature of $A$-bisimulations for
the case of perfect information \cite[Theorem 6]{AlurHKV98}.

\subsection{Discussion
} \label{sec:comparison}

In~\cite{BelardinelliCDJJ17} we introduced a different notion of
bisimulation for the ATL fragment only of ATL$^*$. Specifically, the
original Def.~6 of bisimulation in~\cite{BelardinelliCDJJ17} differs
from Def.~\ref{def:bisim} in that it lacks condition~\ref{it:inject}
(we refer to the paper for further details). Unfortunately, this
weaker notion of bisimulation
makes the preservation theorem false.
In this paper we remedy the problem by requiring the injectiveness of the bisimulation relation w.r.t.~common knowledge neighbourhoods, i.e., by including
condition~\ref{it:inject} in Def.~\ref{def:bisim}.

To illustrate the issue consider the single-agent \mkCEGS with states,
actions, and transitions as illustrated in Fig.~\ref{fig:vietrictex},
and a relation $\faltbisim[A]$ such that $q_{i}\faltbisim[A]q'_{i}$
for $i\in\{0,3\}$ and $q_{i}\faltbisim[A]q'_{ij}$ for $i,j\in\{1,2\}$,
as indicated by the red lines in Fig.~\ref{fig:vietrictex}.  A
case-by-case analysis shows that $\faltbisim[A]$ satisfies
condition~\ref{it:original} in Def.~\ref{def:bisim}, and therefore the
two iCGS are bisimilar according to the notion presented
in \cite{BelardinelliCDJJ17}.  However, they do not satisfy
condition~\ref{it:inject}. Indeed, we have that $q_{1} \faltbisim[A]
q'_{1j}$, for $j\in\{1,2\}$, but $C_{1}(q'_{11}) \neq C_{1}(q'_{12})$,
and similarly for $q_2$, $q'_{21}$, $q'_{22}$. In particular, the iCGS
$\G_1$ and $\G_2$ do not satisfy the same formulas: we have
$(\G_1,q_0) \not \models\coop{1}Fp$ while
$(\G_2,q_0)\models\coop{1}Fp$.  Indeed, in $\G_2$ a memoryless
strategy for agent $1$ to achieve $F p$ is to choose action $a$ in
knowledge set $\{ q'_{1,1}, q'_{2,1} \}$ and action $b$ in $\{ q'_{1,2},
q'_{2,2} \}$. Such a strategy is not transferable as a memoryless
strategy into $\G_1$, as it would require one bit memory, that is, on
the first visit of knowledge set $\{ q_{1}, q_{2} \}$ agent $1$ would
play $a$, then $b$ on the second visit.
Hence, the notion of bisimulation introduced in \cite{BelardinelliCDJJ17} is too weak even to
preserve formulas in the fragment ATL; while the present
Def.~\ref{def:bisim} does preserve the whole ATL$^*$.

\begin{figure}[t]\centering
\begin{tikzpicture}[>=latex, scale=0.75]
 \tikzstyle{state}=[circle, draw, trans, minimum size=3mm]
\tikzstyle{trans}=[font=\footnotesize]

\node[state] (s0) at (0.0, 0.0) {$q_0$};
\node[state] (s1) at (-2.0, -2.0) {$q_1$};
\node[state] (s2) at (2.0, -2.0) {$q_2$};
\node[state,label={p}] (s3) at (0.0, -4.0) {$q_3$};
             
\path [->,style=solid, semithick, shorten >=1pt, auto, node distance=7cm]
(s0) edge (s1)
(s0) edge (s2)
(s1) edge node [below] {b} (s3)
(s2) edge node [above] {a} (s3)
(s3) edge [loop below] (s3)
(s2) edge [loop right] node [below] {b} (s2)
(s1) edge [loop left] node [below] {a} (s1)
;

\path [-,style=dotted, semithick, shorten >=1pt, auto, node distance=7cm]
(s1) edge (s2)
;

\def\mvr{8}

\node[state] (m0) at (0.0+\mvr, 1.0) {$q'_0$};
\node[state] (m11) at (-2.0+\mvr, -1.0) {$q'_{11}$};
\node[state] (m21) at (2.0+\mvr, -1.0) {$q'_{21}$};
\node[state] (m12) at (-2.0+\mvr, -3.0) {$q'_{12}$};
\node[state] (m22) at (2.0+\mvr, -3.0) {$q'_{22}$};
\node[state,label=right:{p}] (m3) at (0.0+\mvr, -5.0) {$q'_{3}$};

\path [->,style=solid, semithick, shorten >=1pt, auto, node distance=7cm]
(m0) edge (m11)
(m0) edge (m21)
(m11) edge [bend left] node {a} (m12)
(m12) edge [bend left] node {a} (m11)
(m21) edge [bend left] node {b} (m22)
(m22) edge [bend left] node {b} (m21)
(m11) edge [bend left] node [left] {b}  ([xshift=-0.4em,yshift=1.2em]m3)
(m21) edge [bend right] node [right] {a} ([xshift=0.4em,yshift=1.2em]m3)
(m12) edge node [left] {b} (m3)
(m22) edge node [right] {a} (m3)
(m3) edge [loop below] (m3)
;

\path [-,style=dotted, semithick, shorten >=1pt, auto, node distance=7cm]
(m11) edge (m21)
(m12) edge (m22)
;

\path [-,style=solid, semithick, shorten >=1pt, auto, node distance=7cm,draw=red]
(s0) edge (m0)
(s1) edge (m11)
(s1) edge (m12)
(s2) edge (m21)
(s2) edge (m22)
(s3) edge (m3)
;

\node[label] (mod0) at (0.0,2.5) {$\G_1$};
\node[label] (mod0) at (0.0+\mvr,2.5) {$\G_2$};

 \end{tikzpicture}
\caption{two iCGS satisfying only condition~1 in Def.~\ref{def:bisim}.}\label{fig:vietrictex}
\end{figure}
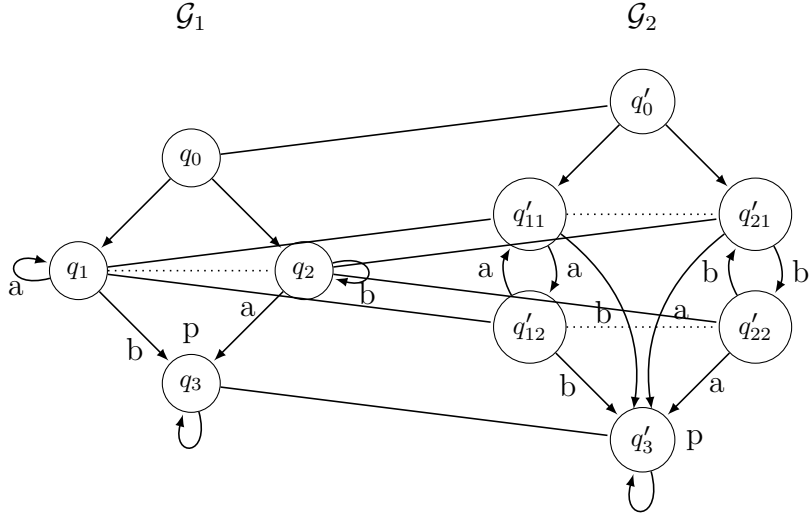

\newcommand{\existsbf}{\pmb{\exists}}
\newcommand{\forallbf}{\pmb{\forall}}

\subsection{Computational Complexity}\label{sec:complexity}

Theorem~\ref{thm-bisim} shows that two $A$-bisimilar iCGS are
equivalent with respect to coalition $A$'s strategic abilities.  This
can be useful when verifying abilities in realistic multi-agent
systems. In such cases, the iCGS typically arises through some kind of
product of local components, and suffers from state-space explosion as
well as transition-space explosion.  If there is a smaller bisimilar
iCGS, one can use the latter as the input to the model checking
procedure.  Since the hardness of the model checking problem for ATL$^*$
with imperfect information is mostly related to the size of the model,
using an equivalent reduced iCGS can turn an impossible task into a
feasible one.  An interesting question is therefore if one can
automatically synthesise such reduced models, or at least check
whether two given iCGS are bisimilar.  We briefly investigate the
issue hereafter. We call a pair $(\G, s)$, for $s$ in $\G$, a {\em
pointed model}.

\begin{theorem}\label{prop:upperbound}
Checking $A$-bisimilarity of two pointed iCGS is in \Sigmacomplx{5}.
\end{theorem}
\begin{proof}Complexity-wise, the most problematic part of Def.~\ref{def:bisim} is
the quantification over strategy simulators $ST$. Being a mapping from
partial strategies to partial strategies, $ST$ has exponential size
w.r.t.~the size of the model, which suggests at least exponential time
for verifying bisimilarity.  Fortunately, a closer look at the
definition of strategy simulators, and at condition 2 in
Def.~\ref{def:bisim}, reveals that $ST$ can be split into local
mappings from mutually disjoint common knowledge neighbourhoods
$C_A$ in $\G$ to unique common knowledge neighbourhoods
$C_A'$ in $\G'$.  Thus, the conceptual structure of
Def.~\ref{def:bisim} can be (with a slight abuse of notation)
summarized as:
\begin{small}
\begin{eqnarray*}
bisimilar((\G,q),(\G',q')) & \text{iff} & \existsbf(\altbisim{A})\ \ simul((\G,q),(\G',q'),\altbisim{A}) \\
  && \quad \land\ simul((\G',q'),(\G,q),\altbisim{A})\ \land\ q\altbisim{A}q', \\
simul((\G,q),(\G',q'),\altbisim{A}) & \text{iff} & \forallbf(\mathcal{C}_A\in\G')\ \forallbf(\sigma_A\in\mathcal{C}_A)\existsbf(\sigma_A'\in\mathcal{C}_A') \\
  && \forallbf(\hat{q}\in\mathcal{C}_A, \hat{q}'\in\mathcal{C}_A', \hat{q}\altbisim{A}\hat{q}')\ \ match(\hat{q},\hat{q}') \\
  && \quad \land\ simulepist(\hat{q},\hat{q}')\ \land\ simultrans(\mathcal{C}_A,\mathcal{C}_A'), \\
match(\hat{q},\hat{q}') & \text{iff} & \pi(\hat{q}) = \pi(\hat{q}'), \\
simulepist(\hat{q},\hat{q}') & \text{iff} & \forallbf(i\in A)\ \forallbf(\hat{q}'\sim_i r')\ \existsbf(r)\ \ \hat{q}\sim_i r, \ r \altbisim{A} r' \\
simultrans(\mathcal{C}_A,\mathcal{C}_A') & \text{iff} & \forallbf(r\in\mathcal{C}_A, r'\in\mathcal{C}_A', r\altbisim{A}r')\ \forallbf(s'\in succ(r',\sigma_A')) \\
  && \existsbf(s\in succ(r,\sigma_A))\ \ s\altbisim{A}s'.
\end{eqnarray*}
\end{small}
Notice that all the structures quantified in the above expressions are of polynomial size with respect to the number of states, transitions, and epistemic links in $\G$ and $\G'$.
Thus, it is relatively straightforward to observe, by a suitable translation to the QBF problem (i.e., satisfiability of Quantified Boolean Formulae) that:
\begin{enumerate2}
\item checking $simulepist(\hat{q},\hat{q}')$ and $simultrans(\mathcal{C}_A,\mathcal{C}_A')$ is in \Picomplx{2},
\item checking $simul((\G,q),(\G',q'),\altbisim{A})$ is in \Picomplx{4},
\item and thus, finally, checking $bisimilar((\G,q),(\G',q'))$ is in \Sigmacomplx{5}.
\end{enumerate2}
\end{proof}

We suspect that the upper bound is tight, as we do not see how any of the quantifier alternations, included in the above QBF translation, could be collapsed.
For the moment, however, we only show the following, and leave the question of the exact complexity for future work.

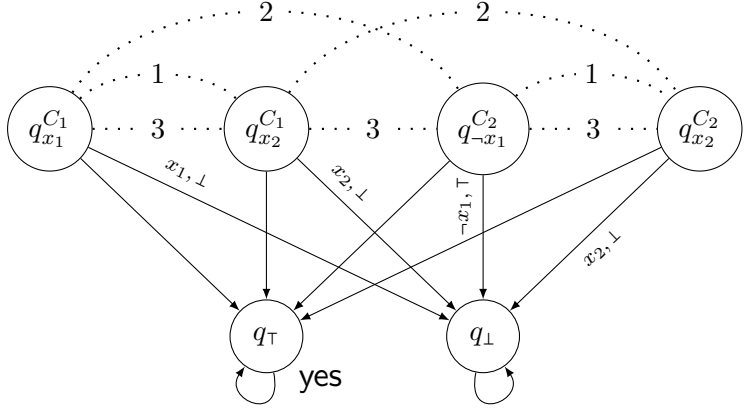
\begin{figure}[t]\centering
\begin{tikzpicture}[scale=1.2]

    \tikzstyle{astate}=[circle,draw,trans]
    \tikzstyle{trans}=[font=\footnotesize]

    \path   (-3.6,0) node[astate] (q11) {\small$q^{C_1}_{x_1}$}
            (-1.2,0) node[astate] (q12) {\small$q^{C_1}_{x_2}$}
            (1.2,0) node[astate] (q21) {\small$q^{C_2}_{\neg x_1}$}
            (3.6,0) node[astate] (q22) {\small$q^{C_2}_{x_2}$}

            (-1.2,-2.3) node[astate] (qt) {\small${\;}q_\top{\;}$} +(+0.6,-0.5) node[transform shape] {\small $\prop{yes}$}
            (1.2,-2.3) node[astate] (qb) {\small${\;}q_\bot{\;}$}

    ;

    \draw[loosely dotted,thick](q11) ..controls +(+0.7,+0.7) and +(-0.7,+0.7).. (q12) node[midway,sloped] {\small\colorbox{white}{$1$}};
    \draw[loosely dotted,thick](q21) ..controls +(+0.7,+0.7) and +(-0.7,+0.7).. (q22) node[midway,sloped] {\small\colorbox{white}{$1$}};

    \draw[loosely dotted,thick](q11) ..controls +(+1,+1.6) and +(-1,+1.6).. (q21) node[midway,sloped] {\small\colorbox{white}{$2$}};
    \draw[loosely dotted,thick](q12) ..controls +(+1,+1.6) and +(-1,+1.6).. (q22) node[midway,sloped] {\small\colorbox{white}{$2$}};

    \draw[loosely dotted,thick](q11)--(q12) node[midway,sloped] {\small\colorbox{white}{$3$}};
    \draw[loosely dotted,thick](q12)--(q21) node[midway,sloped] {\small\colorbox{white}{$3$}};
    \draw[loosely dotted,thick](q21)--(q22) node[midway,sloped] {\small\colorbox{white}{$3$}};

    \draw[-latex](q11)--(qt);
    \draw[-latex](q12)--(qt);
    \draw[-latex](q21)--(qt);
    \draw[-latex](q22)--(qt);

    \draw[-latex](q11)--(qb) node[near start,sloped,above] {\scriptsize $x_1,\bot$};
    \draw[-latex](q12)--(qb) node[near start,sloped,above] {\scriptsize $x_2,\bot$};
    \draw[-latex](q21)--(qb) node[near start,sloped,above] {\scriptsize $\neg x_1,\top$};
    \draw[-latex](q22)--(qb) node[midway,sloped,below] {\scriptsize $x_2,\bot$};

    \draw[-latex](qt) ..controls +(0.2,-1) and +(-0.5,-0.7).. (qt);
    \draw[-latex](qb) ..controls +(-0.2,-1) and +(0.5,-0.7).. (qb);
 \end{tikzpicture}
\caption{Model $M_\Phi$ for $\Phi\equiv C_1\land C_2$ with $C_1\equiv x_1\lor x_2$ and $C_2\equiv\neg x_1\lor x_2$. Only transitions leading to $q_\bot$ are labeled; the other combinations of actions lead to $q_\top$.}\label{fig:reduction}
\end{figure}

\begin{theorem}\label{prop:lowerbound}
Checking $A$-bisimilarity of two pointed iCGS is \NP-hard.
\end{theorem}
\begin{proof}
We adapt the SAT reduction from~\cite[Proposition~11]{Bulling11mu-ijcai}.
Let us first recall the idea of that reduction.  Given a Boolean
formula $\Phi$ in CNF, we build a $3$-agent iCGS $\G_\Phi$ where each
literal $l$ from clause $\psi$ in $\Phi$ is associated with a state
$q^\psi_l$.
Each agent has a different role in the construction: player 1 selects the literals to be satisfied (one literal per clause in $\Phi$), and player 2 chooses the truth values for those literals. Player 3 ``simulates'' the common knowledge neighborhood for $\{1,2\}$, and thus ensures that a successful assignment must work for all the clauses in $\Phi$.

Formally, at state $q^\psi_l$, player $1$ indicates a literal from
$\psi$, and player $2$ decides on the valuation of the underlying
Boolean variable. If $1$ indicated a ``wrong'' literal $l'\neq l$ then
the system proceeds to state $q_\top$ where proposition $\prop{yes}$
holds. The same happens if $1$ indicated the ``right'' literal ($l$)
and $2$ selected the valuation that makes $l$ true. Otherwise the
system proceeds to the ``sink'' state $q_\bot$.
Player $1$ must select literals uniformly within clauses, which is
imposed by fixing $q^\psi_{l}\sim_1 q^{\psi'}_{l'}$ iff
$\psi=\psi'$. Player $2$ is to select uniform valuations of variables,
i.e., $q^\psi_{l}\sim_2 q^{\psi'}_{l'}$ iff $var(l)=var(l')$, where
$var(l)$ is the variable contained in $l$. Finally, all states except
$q_\top,q_\bot$ are indistinguishable for agent $3$. An example of the
construction is presented in Figure~\ref{fig:reduction}.
Let $q_0$ be an arbitrary ``literal'' state, e.g., the one for the
first literal in the first clause.  Then, $\mathbf{SAT}(\Phi)$ iff
$(\G_\Phi,q_0) \models \coop{1,2,3}\Next\prop{yes}$ according to the
subjective semantics of \ATL with imperfect information and imperfect
recall.

Now, we construct iCGS $\G_\Phi'$ by adding action $skip$, available
to agent $1$ at every ``literal'' state, which always enforces a
transition to $q_\top$ regardless of the actions selected by agents 2
and 3.  In particular, $\G_\Phi'$ adds to $\G_\Phi$ a strategy for
coalition $\{1,2,3\}$ whose only successor state from all the
``literal'' states is $q_\top$.  Thus, $(\G_\Phi,q_0) \models
\coop{1,2,3}\Next\prop{yes}$ iff $(\G_\Phi,q_0) \altbisim{\{1,2,3\}}
(\G_\Phi',q_0)$, which completes the reduction.
\end{proof}

In sum, the problem of deciding bisimilarity is easier than exponential time, but still far from practically automatizable.
To obtain a reduction, one must propose the reduced model and the bisimulation according to one's intuition, and verify correctness of the bisimulation by hand.
We will show an extensive example of how this can be done in Section~\ref{3ballot}.

\section{Towards the Hennessy-Milner Property} \label{sec:HM}

In this section we show that the notion of bisimulation we introduced
in Section~\ref{bisimulation} does not enjoy the Hennessy-Milner (HM)
property, that is, some iCGS are logically equivalent (i.e., they
satisfy the same formulas in ATL$^*$) and yet they are not bisimilar.
This is the case for both the subjective and objective variants of our
semantics. More precisely, for $x \in \{sub, obj \}$, we say that
states $q \in \G$ and $q' \in \G'$ are {\em $A$-equivalent}, or
$q \mkbisxa q'$, iff for every $A$-formula $\phi$, $(\G,
q) \models_{x} \phi$ iff $(\G, q') \models_{x} \phi$.  The iCGS
$\G, \G'$ will be omitted as clear from the context.
In other words, relation $\mkbisxa$ connects those states of two iCGS
that satisfy the same formulas in ATL$^*$ that refer to coalition $A$.

In many logics, such as the Hennessy-Milner logic, CTL, CTL$^*$, ATL
and ATL$^*$ interpreted under perfect information \cite{AlurHKV98},
logical equivalence can be characterised in a local and efficient way
by means of bisimulations. By contrast, the bisimulations for
ATL$^*$ introduced in this paper are strictly weaker than logical
equivalence $\mkbisxa$. In particular, in Section~\ref{sec:limitations} we show
that there are iCGS that satisfy the same formulas in ATL$^*$, while
not being bisimilar.
Then, in Section~\ref{sec:tight} we present a necessary local
condition that needs to be satisfied by $\mkbissa$.  As we show, the
condition is not sufficient however. Nevertheless, it is an important
piece of the currently partially known puzzle about the
characterisation power of bisimulations for ATL$^*$ under imperfect
information.

In the rest of the section let $T \subseteq S$ and $R \subseteq S
\times S'$. We define \emph{the image} of $T$ w.r.t.~$R$:
\begin{align*}
  \mkimg{T}{R} & = \{\mkstate \in S' \mkmrule \text{for some } \mkstate' \in T \;(\mkstate',\mkstate)\in R \}.
\end{align*}
For each $\mkstate \in S$, let $E_A(\mkstate)$ denote the set
of all states that are indistinguishable from $\mkstate$ according to
the ``{\em everybody in $A$ knows}'' relation $\sim_A^E$.

\subsection{Failure of the HM Property} \label{sec:limitations}

We immediately state the failure of the HM property for our notion of
bisimulation for ATL$^*$ under imperfect information.
\begin{theorem} \label{HM}
For $x \in \{sub, obj \}$, there exists iCGS $\G, \G'$, and states $q,
q'$, such that $q \mkbisx{Ag} q'$ and $q \notaltbisim{Ag} q'$.
\end{theorem}

We prove Theorem~\ref{HM} first for the subjective and then for the
objective semantics.

\begin{figure}[t]\centering
\begin{tabular}{ccc}
 $\G_3$ & \qquad\qquad & $\G_4$ \\ \\
\begin{tikzpicture}[>=latex, transform shape, scale = 0.75]
\tikzstyle{initstate}=[circle,draw,trans, minimum size=8mm, fill=yellow]
\tikzstyle{state}=[circle,draw,trans, minimum size=8mm]
\tikzstyle{trans}=[font=\footnotesize]

\def\stpd{3.5}

\node[initstate] (m0) at (-1.5, -0.5*\stpd) {$q_1$};
\node[state] (m1) at (0.0, -\stpd) {$q_2$};
\node[state] (m2) at (2.5, -\stpd) {$q_3$};
\node[state] (m3) at (4.0, -0.5*\stpd) {$q_4$};

\node[state, label=-90:{$p$}] (fin) at (1.5, -2*\stpd) {$q_\top$};
\node[state, label=-90:{}] (sink) at (-1.5, -2.4*\stpd) {$q_\bot$};

\path [-,style=dotted,shorten >=1pt, auto, node distance=7cm, semithick]

(m0) edge node {1} (m1)
(m1) edge node {2} (m2)
(m2) edge node {1} (m3)
;

\path [->,style=solid,shorten >=1pt, auto, node distance=7cm, semithick]

(m0) edge [bend right] node[left,text width = 0.8cm] {$(a,x)$ $(b,y)$} (fin)
(m1) edge [bend right] node[midway,right,text width = 0.8cm] {$(a,x)$ $(b,y)$} (fin)
(m2) edge [bend left] node[midway,left,text width = 0.8cm] {$(a,x)$ $(b,y)$} (fin)
(m3) edge [bend left] node[right,text width = 0.8cm] {$(a,x)$ $(b,y)$} (fin)

(sink) edge [loop left] (sink)

;
 \end{tikzpicture}
 & &  
\begin{tikzpicture}[>=latex, transform shape, scale = 0.75]
\tikzstyle{initstate}=[circle,draw,trans, minimum size=8mm, fill=yellow]
\tikzstyle{state}=[circle,draw,trans, minimum size=8mm]
\tikzstyle{trans}=[font=\footnotesize]

\def\stpd{3.5}

\node[initstate] (m0) at (-1.5, -0.5*\stpd) {$q'_1$};
\node[state] (m1) at (0.0, -\stpd) {$q'_2$};
\node[state] (m2) at (2.5, -\stpd) {$q'_3$};
\node[state] (m3) at (4.0, -0.5*\stpd) {$q'_4$};

\node[state, label=-90:{$p$}] (fin) at (1.5, -2*\stpd) {$q'_\top$};
\node[state, label=-90:{}] (sink) at (-1.5, -2.4*\stpd) {$q'_\bot$};

\path [-,style=dotted,shorten >=1pt, auto, node distance=7cm, semithick]

(m0) edge node {1} (m1)
(m1) edge node {2} (m2)
(m2) edge node {1} (m3)
;

\path [->,style=solid,shorten >=1pt, auto, node distance=7cm, semithick]

(m0) edge [bend right] node[left] {$(a,x)$} (fin)
(m1) edge [bend right] node[midway,right,text width = 0.8cm] {$(a,x)$ $(b,y)$} (fin)
(m2) edge [bend left] node[midway,left,text width = 0.8cm] {$(a,x)$ $(b,y)$} (fin)
(m3) edge [bend left] node[right] {$(b,y)$} (fin)

(sink) edge [loop left] (sink)

;
 \end{tikzpicture}
\end{tabular}
\caption{A counterexample to Hennessy-Milner property for subjective semantics. Available actions: $\{a,b,c\}$ for agent 1, $\{x,y,z\}$ for agent 2.}
\label{fig:indirect-conflict-2}
\end{figure}
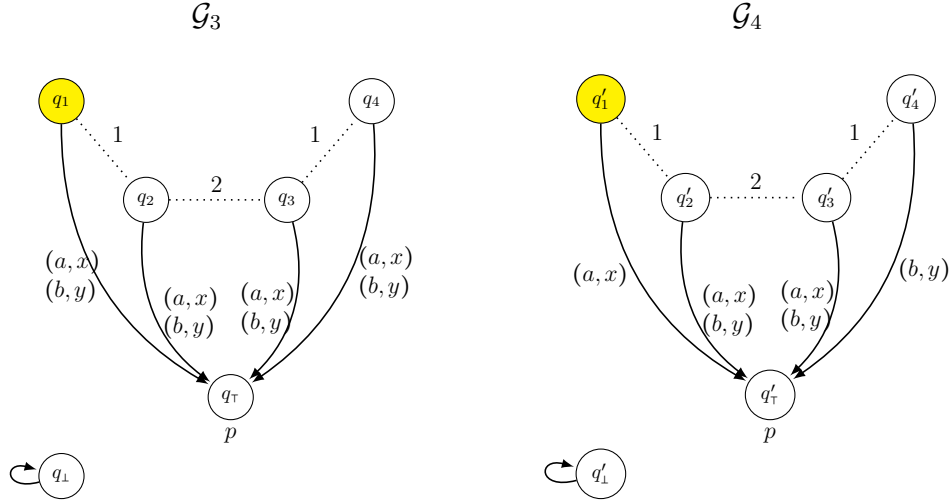

\para{
Subjective Semantics.} Consider the two-agent \mkCEGS\ in
Fig.~\ref{fig:indirect-conflict-2}.  In each state agent 1 can execute
actions $\{a,b,c\}$ while agent 2 can execute $\{x,y,z\}$. The
transitions shown lead to $q_{\top}$ and $q'_{\top}$, while the
omitted transitions lead to $q_{\bot}$ and $q'_{\bot}$, respectively.
Via case-by-case analysis it can be proved that $q_i\mkbiss{Ag} q_j$,
$q'_i\mkbiss{Ag} q'_j$, and $q_i\mkbiss{Ag} q'_j$, for
$i,j\in\{1,2,3,4\}$, and also $q_\bot \mkbiss{Ag} q'_\bot$, $q_\top
\mkbiss{Ag} q'_\top$.  
Therefore, $\G_3$ and $\G_4$ are $Ag$-equivalent, i.e., they satisfy the
same $Ag$-formulas in ATL$^*$.  However, we show that there is no
$Ag$-bisimulation between the two iGCS. In particular, for any $i, j
\in \{1,2,3,4\}$, state $q_i$ cannot be $Ag$-bisimular with any state
$q'_j$.  In the proof we make use of the following lemma.
\begin{lemma}\label{lemma:preserveck}
 Let $A \subseteq Ag$. If $q\altbisim{A}q'$ then
\begin{enumerate}
\item   for all $r'\in C_A(q')$, there exists $r\in C_A(q)$ such that $r \altsim{A} r'$;
\item  for all $r\in C_A(q)$, there exists $r'\in C_A(q')$ such that $r \altsim{A} r'$.  
\end{enumerate}
\end{lemma}
\begin{proof}
  The thesis of the lemma follows immediately from the observation
  that $C_A(q')\subseteq\mkimg{C_A(q)}{\altsim{A}}$, which in turn can
  be shown via straightforward induction on the length of the path
  w.r.t.~relation $\sim^C_A$ that joins $q'$
with $r'$.
\end{proof}

Now, set $T = \{q_1,q_2,q_3,q_4\}$ and $T' = \{q_1',q_2',q_3',q_4'\}$,
and note that, for the partial strategy $\sigma_{Ag}$ whose
domain is $T$ and is defined by $\sigma_{Ag}(q_i) = (a,x)$ for all
$q_i \in T$, we have $succ(T,\sigma_{Ag}) = \{ q_\top\}$.  On
the other hand, note that for any partial strategy
$\sigma'_{Ag}$ whose domain is $T'$, we have 
$succ(T',\sigma'_{Ag}) = \{q'_\top,q'_\bot\}$.

So, if we assume the existence of a strategy simulator 
$ST : \mkpartstrata{C_A(q_i)} \rightarrow \mkpartstrata{C_A'(q'_i)}$,
then $succ(T',ST(\sigma^1_{Ag})) = \{q'_\bot,q'_\top\}$,
which implies that 
there exists $q'_j \in T'$ such that $succ(q'_j, \sigma'_{Ag}) = q'_\bot$.
On the other hand, by Lemma~\ref{lemma:preserveck}, there exist $q_i \in T$ such that 
$q_i \altsim{Ag} q'_j$.
But $succ(q_i,\sigma^1_{Ag}) = q_\top$, which is in contradiction
with point $1.(c)$ in Def.~\ref{def:bisim}, as it cannot be the case that 
$q_\top \altsim{Ag} q'_\bot$.

\begin{figure}[t]\centering
\begin{tabular}{ccc}
 $\G_5$ & \qquad\qquad & $\G_6$
\\ \\
\begin{tikzpicture}[>=latex,scale=1.2]
 \tikzstyle{initstate}=[circle,draw,trans, minimum size=8mm, fill=yellow]
\tikzstyle{astate}=[circle,draw,trans, minimum size=8mm]
\tikzstyle{trans}=[font=\footnotesize]

    \path   (-1.2,0) node[initstate] (s0) {\small${\;}q_0{\;}$}
            (1.2,0) node[astate] (s1) {\small${\;}q_1{\;}$}

            (-1.2,-2) node[astate] (s2) {\small${\;}q_2{\;}$} +(+0.5,-0.5) node[transform shape] {\small $\prop{win}$}
            (1.2,-2) node[astate] (s3) {\small${\;}q_3{\;}$}

    ;

    \draw[loosely dotted,thick](s0)--(s1) node[midway,sloped] {\small\colorbox{white}{$1$}};

    \draw[-latex](s0)--(s2) node[midway,left] {\scriptsize$L$};
    \draw[-latex](s0)--(s3) node[midway,sloped,below,near start] {\scriptsize$R$};
    \draw[-latex](s1)--(s2) node[midway,sloped,below,near start] {\scriptsize$R$};
    \draw[-latex](s1)--(s3) node[midway,right] {\scriptsize$L$};

    \draw[-latex](s2) ..controls +(0.2,-1) and +(-0.5,-0.7).. (s2);
    \draw[-latex](s3) ..controls +(-0.2,-1) and +(0.5,-0.7).. (s3);
 \end{tikzpicture}
 & &
\begin{tikzpicture}[>=latex,scale=1.2]
 \tikzstyle{initstate}=[circle,draw,trans, minimum size=8mm, fill=yellow]
\tikzstyle{astate}=[circle,draw,trans, minimum size=8mm]
\tikzstyle{trans}=[font=\footnotesize]

    \path   (-1.2,0) node[initstate] (q0) {\small${\;}q'_0{\;}$}

            (-1.2,-2) node[astate] (q2) {\small${\;}q'_2{\;}$} +(+0.5,-0.5) node[transform shape] {\small $\prop{win}$}
            (1.2,-2) node[astate] (q3) {\small${\;}q'_3{\;}$}

    ;

    \draw[-latex](q0)--(q2) node[midway,left] {\scriptsize$L$};
    \draw[-latex](q0)--(q3) node[midway,sloped,below,near start] {\scriptsize$R$};

    \draw[-latex](q2) ..controls +(0.2,-1) and +(-0.5,-0.7).. (q2);
    \draw[-latex](q3) ..controls +(-0.2,-1) and +(0.5,-0.7).. (q3);
 \end{tikzpicture}
\end{tabular}
\caption{A counterexample to the Hennessy-Milner property for objective semantics.}
\label{fig:poorduck}
\end{figure}
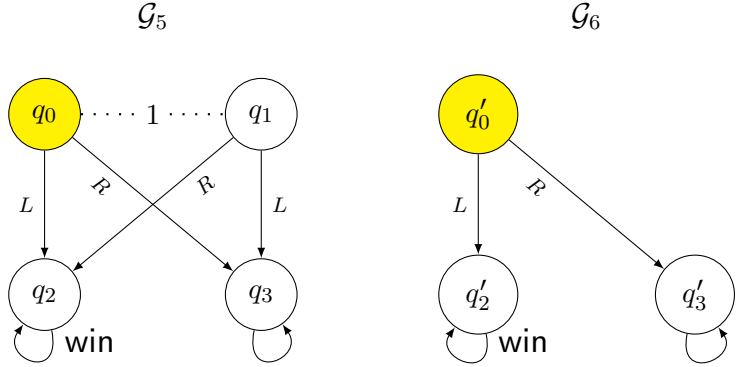

\para{
Objective Semantics.}  Consider the single-agent \mkCEGS\ in
Fig.~\ref{fig:poorduck} in which $q_0 \nmkbissag q'_0$.  Observe first
that, by symmetry, we have $q_0 \mkbis{Ag}^x q_1$, for
$x\in\{\mathit{subj},\mathit{obj}\}$.  Hence, the same two strategies
(play $R$, play $L$) can be executed in $q_0$ and $q_1$.  For each
such strategy, its subtree starting from $q_0$ is isomorphic to the
corresponding tree starting from $q'_0$, we thus have $q_0 \mkbiso{Ag}
q'_0$.  In particular, we can establish that $q_{i} \mkbiso{Ag}
q'_{i}$, for all $i\in\{0,2,3\}$.
On the other hand,
as both $q_{0}\mkbiso{Ag} q'_{0}$ and $q_{1}\mkbiso{Ag} q'_{0}$,
the simulator function $ST_{\{q_0, q_1\},\{q'_0\}}$ should satisfy
$ST_{\{q_0, q_1\},\{q'_0\}}(L) = L$ and $ST_{\{q_0, q_1\},\{q'_0\}}(L) = R$, a contradiction.
Hence, it is not the case that $q_0{\altbisim{Ag}}q'_0$.

~\\

As a consequence of Theorem~\ref{HM}, our notion of bisimulation is
sufficient for the preservation of formulas in ATL$^*$, but it is not
necessary. In particular, it does not enjoy the Hennessy-Milner property. 
In the following section we investigate this problem more closely.

\subsection{Towards a Tight Characterization of ATL$^*$-Equivalence} \label{sec:tight}

In this section we investigate necessary, rather than sufficient,
conditions for logical equivalence.  Specifically, we introduce a
notion of pre-simulation and prove that it is implied by logical
equivalence for the subjective variant of the semantics, while being
strictly weaker.
\begin{definition}[Pre-bisimulation]\label{def:presim}
Let $\G$
and $\G'$
be two iCGS defined on the same sets $Ag$ of agents and $AP$ of atoms.
Let $A\subseteq Ag$ be a coalition of agents.
  A relation $\mkpaltsima \subseteq S \times S'$
  is a \emph{pre-simulation} iff
there exists a simulator $ST$ of partial strategies for $A$ w.r.t.~$\mkpaltsima$,
such that $q \mkpaltsima q'$ implies that conditions (a) and (b) in Def.~\ref{def:bisim} hold and
\begin{enumerate}

\item[(c')] For every states $r \in E_A(q)$, $r' \in E_A'(q')$ such that $r \mkpaltsima r'$,
  for every partial uniform strategy $\sigma_A \in PStr_A(E_A(q))$,
  and every state $s' \in succ(r', ST(\sigma_A))$,
  there exists a state $s\in succ(r, \sigma_A)$ such that $s \mkpaltsima s'$.
\end{enumerate}

A relation $\mkpaltbisima$ is called \emph{pre-bisimulation} if both $\mkpaltbisima$ and
$\mkpaltbisima^{-1}$ are \emph{pre-simulations}.
\end{definition}

The key difference between the pre-bisimulations in
Def.~\ref{def:presim} and the bisimulations in Def.~\ref{def:bisim}
regards the domains of the partial uniform strategies mapped by
strategy simulator $ST$: these are common knowledge neighbourhoods for
bisimulation, and simple collective knowledge neighbourhoods for
pre-bisimulations. Moreover, pre-bisimulation need not to satisfy condition (2) in Def.~\ref{def:bisim}.

We can now prove the following result.
\begin{theorem}[Necessary Conditions on $\mkbissa$]\label{theorem:necondsubj}
\mbox{}

\noindent
If $\mkstate\mkbissa\mkstate'$, then there exists a pre-bisimulation $\mkpaltbisima$ such that $\mkstate\mkpaltbisima\mkstate'$.
\end{theorem}
\begin{proof}
  We prove that $\mkbissa$ is a pre-bisimulation.  The proof is by
  contradiction. It is routine to prove that $\mkbissa$ satisfies
  points~(a) and~(b) in Def.~\ref{def:bisim}, we therefore focus
  on point~(c').  Let $\mkstate \mkbissa \mkstate'$ and to obtain a contradiction suppose that $\mkstrata\in\mkpartstrata{E_A(q)}$ is such that for every
  $\mkstrata'\in \mkpartstrata{E_A(q')}$ there exists $\mkstateb \in E_A(q)$, $\mkstateb' \in E_A(q')$, and $s' \in succ(r', \mkstrata')$ such that for every $s \in
  succ(r,\mkstrata)$, we have $s \nmkbissa s'$.
This in particular means that for each such $s$ there exists a formula
$\phi_{s}^{\mkstrata'}$ in ATL$^*$ such that $(\G, s)
\mkmodelss\phi_{s}^{\mkstrata'}$ and $(\G', s')
\mkmodelss\neg\phi_{s}^{\mkstrata'}$. Now, define formula
$\phi_r^{\mkstrata'} = \bigvee_{s \in succ(r,
  \mkstrata)}\phi_{s}^{\mkstrata'}$ and observe that
$\phi_r^{\mkstrata'}$ is true in the next step from $r$ by using
$\mkstrata$, whereas this is not the case for $r'$.
Finally, let $\phi =
\bigwedge_{\mkstrata'\in\mkpartstrata{E_A(q')}}\phi_r^{\mkstrata'}$.  To
conclude, it suffices to observe that $(\G, \mkstate)
\mkmodelss\coop{A}X\phi$ and $(\G', \mkstate')
\mkmodelss\neg\coop{A}X\phi$, a contradiction with
$\mkstate\mkbissa\mkstate'$.
\end{proof}

The following corollary to Theorem~\ref{theorem:necondsubj} deals with the case of objective semantics.

\begin{corollary}[Necessary Condition on $\mkbisoa$]\label{theorem:necondobj}
  If $\mkstate\mkbisoa\mkstate'$, then for each uniform strategy $\mkstrata$
  there exists a uniform strategy $\mkstrata'$ such that:
  \[
succ(q',\mkstrata') \subseteq  \mkimg{succ(q,\mkstrata)}{\mkbisoa}
  \]
\end{corollary}

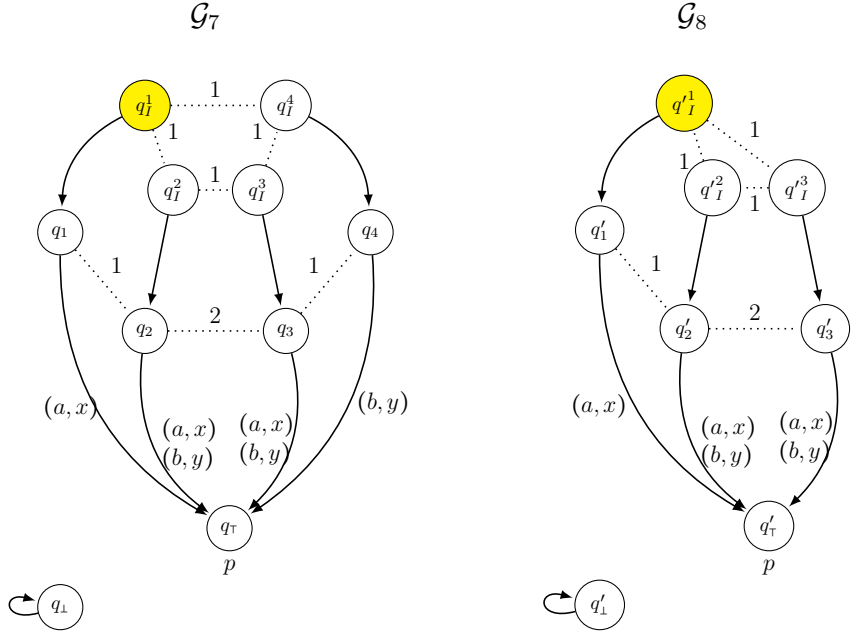
\begin{figure}[t]\centering
\begin{tabular}{ccc}
 $\G_7$ & \qquad\qquad & $\G_8$
\\ \\
\begin{tikzpicture}[>=latex, transform shape, scale=0.75]
 \tikzstyle{initstate}=[circle,draw,trans, minimum size=8mm, fill=yellow]
\tikzstyle{state}=[circle,draw,trans, minimum size=8mm]
\tikzstyle{trans}=[font=\footnotesize]

\def\stpd{3.5}

\node[initstate] (s0) at (0.0, 0.5) {$q_I^1$};
\node[state] (s1) at (2.5, 0.5) {$q_I^4$};
\node[state] (s2) at (0.5, -1.0) {$q_I^2$};
\node[state] (s3) at (2.0, -1.0) {$q_I^3$};

\node[state] (m0) at (-1.5, -0.5*\stpd) {$q_1$};
\node[state] (m1) at (0.0, -\stpd) {$q_2$};
\node[state] (m2) at (2.5, -\stpd) {$q_3$};
\node[state] (m3) at (4.0, -0.5*\stpd) {$q_4$};

\node[state, label=-90:{$p$}] (fin) at (1.5, -2*\stpd) {$q_\top$};
\node[state, label=-90:{}] (sink) at (-1.5, -2.4*\stpd) {$q_\bot$};

\path [-,style=dotted,shorten >=1pt, auto, node distance=7cm, semithick]
(s0) edge node {1} (s1)
(s0) edge node {1} (s2)
(s3) edge node {1} (s1)
(s2) edge node {1} (s3)

(m0) edge node {1} (m1)
(m1) edge node {2} (m2)
(m2) edge node {1} (m3)
;

\path [->,style=solid,shorten >=1pt, auto, node distance=7cm, semithick]
(s0) edge [bend right] (m0)
(s1) edge [bend left] (m3)
(s2) edge (m1)
(s3) edge (m2)

(m0) edge [bend right] node[left,text width = 0.8cm] {$(a,x)$} (fin)
(m1) edge [bend right] node[midway,right,text width = 0.8cm] {$(a,x)$ $(b,y)$} (fin)
(m2) edge [bend left] node[midway,left,text width = 0.8cm] {$(a,x)$ $(b,y)$} (fin)
(m3) edge [bend left] node[right,text width = 0.8cm] {$(b,y)$} (fin)

(sink) edge [loop left] (sink)

;
 \end{tikzpicture}
 & &
\begin{tikzpicture}[>=latex, transform shape, scale=0.75]
 \tikzstyle{initstate}=[circle,draw,trans, minimum size=8mm, fill=yellow]
\tikzstyle{state}=[circle,draw,trans, minimum size=8mm]
\tikzstyle{trans}=[font=\footnotesize]

\def\stpd{3.5}

\node[initstate] (s0) at (0.0, 0.5) {${q'}^1_I$};
\node[state] (s1) at (2.0, -1.0) {${q'}^3_I$};
\node[state] (s2) at (0.5, -1.0) {${q'}^2_I$};

\node[state] (m0) at (-1.5, -0.5*\stpd) {$q'_1$};
\node[state] (m1) at (0.0, -\stpd) {$q'_2$};
\node[state] (m2) at (2.5, -\stpd) {$q'_3$};

\node[state, label=-90:{$p$}] (fin) at (1.5, -2*\stpd) {$q'_\top$};
\node[state, label=-90:{}] (sink) at (-1.5, -2.4*\stpd) {$q'_\bot$};

\path [-,style=dotted,shorten >=1pt, auto, node distance=7cm, semithick]
(s0) edge node {1} (s1)
(s2) edge node {1} (s0)
(s1) edge node {1} (s2)

(m0) edge node {1} (m1)
(m1) edge node {2} (m2)
;

\path [->,style=solid,shorten >=1pt, auto, node distance=7cm, semithick]
(s0) edge [bend right] (m0)
(s2) edge (m1)
(s1) edge (m2)

(m0) edge [bend right] node[left] {$(a,x)$} (fin)
(m1) edge [bend right] node[midway,right,text width = 0.8cm] {$(a,x)$ $(b,y)$} (fin)
(m2) edge [bend left] node[midway,left,text width = 0.8cm] {$(a,x)$ $(b,y)$} (fin)

(sink) edge [loop left] (sink)

;
 \end{tikzpicture}
\end{tabular}
\caption{Pre-bisimulation does not imply logical equivalence.}
\label{fig:prebisimnotsim}
\end{figure}

It should be noted that pre-bisimilarity does not imply logical
equivalence of two models.  Consider the two-agent \mkCEGS\ in
Fig.~\ref{fig:prebisimnotsim}.  In the states that are
$1$-indistinguishable from $q_I^1$ (respectively, ${q'}^1_I$) the
models are equipped with a single deterministic transition.  In each
of the remaining states agent 1 can execute actions $\{a,b,c\}$ and
agent 2 can execute $\{x,y,z\}$. We assume that the transitions that
are not depicted in the figure lead to $q_{\bot}$ or $q_{\bot}'$,
respectively.  Now, observe that $(\G_7, q_I^1)
\mkmodelss\neg\coop{Ag}Fp$, as it is not possible to choose a
strategies for both agents that select the same pair of actions over
the set of successors of $q_I^1$ and leads to $q_\top$.  On the other
hand, such choice is possible for the successors of ${q'}_I^1$, namely
it suffices to assign $(a,x)$ to each ${q'}^i_I$, where
$i\in\{1,2,3\}$, to see that $(\G_8, {q'}_I^1) \mkmodelss\coop{Ag}Fp$.
Finally, we can define a relation $\mkpaltbisima$ such that
$q_i\mkpaltbisima q'_i$ and $q^i_I\mkpaltbisima {q'}^i_I$ for all
$i\in\{1,2,3 \}$, $q_{\top} \mkpaltbisima q'_{\top}$, $q_{\bot}
\mkpaltbisima q'_{\bot}$, and $q_4\mkpaltbisima q'_3$,
$q^4_I\mkpaltbisima {q'}^3_I$.  A case-by-case check shows that
$\mkpaltbisima$ is a pre-bisimulation.  Specifically, conditions (a)
and (b) in Def.~\ref{def:bisim} are immediate. As regards (c'), notice
that for pre-simulations strategy simulators are defined on collective
knowledge neighbourhoods rather than common knowledge ones. More
specifically, $C_{Ag}(q_1) = \{q_1, q_2, q_3, q_4 \}$ and
$C_{Ag}(q'_1) = \{q_1, q_2, q_3 \}$, while $E_{Ag}(q_1) = \{q_1, q_2,
q_3 \}$ and $E_{Ag}(q'_1) = \{q_1, q_2, q_3 \}$. Then, while a strategy simulator can be defined on $E_{Ag}(q_1)$, $E_{Ag}(q'_1)$, so such simulator exists on $C_{Ag}(q_1)$, $C_{Ag}(q'_1)$.

\section{Case Study: the ThreeBallot Voting Protocol} \label{3ballot}

\begin{figure}[t]\centering

  \newcommand{\wcircle}{\tikz\draw[black,fill=white] (0,0) circle (1ex);}
  \newcommand{\bcircle}{\tikz\draw[black,fill=black] (0,0) circle (1ex);}  
  \setlength{\dashlinedash}{2pt}
  \setlength{\dashlinegap}{4pt}
  \renewcommand{\arraystretch}{1.2}  
  \def \spc {5mm}

  \begin{tabular}{|l c:l c:l c|}
    \hline
    \multicolumn{2}{|c:}{BALLOT} &
    \multicolumn{2}{:c}{BALLOT} &
    \multicolumn{2}{:c|}{BALLOT} \\
    & & & & & \\
    
    \hspace{\spc}Andy Jones \hspace{\spc}& \wcircle\hspace{\spc}
    &
    \hspace{\spc}Andy Jones \hspace{\spc}& \wcircle\hspace{\spc}
    &
    \hspace{\spc}Andy Jones \hspace{\spc}& \bcircle\hspace{\spc}{\ }\\

    \hspace{\spc}Bob Smith \hspace{\spc}& \bcircle\hspace{\spc}
    &
    \hspace{\spc}Bob Smith \hspace{\spc}& \bcircle\hspace{\spc}
    &
    \hspace{\spc}Bob Smith \hspace{\spc}& \wcircle\hspace{\spc}{\ }\\

    \hspace{\spc}Carol Wu \hspace{\spc}& \wcircle\hspace{\spc}
    &
    \hspace{\spc}Carol Wu \hspace{\spc}& \bcircle\hspace{\spc}    
    &
    \hspace{\spc}Carol Wu \hspace{\spc}& \wcircle\hspace{\spc}{\ }\\

    & & & & & \\
    \multicolumn{2}{|c:}{3147524} &
    \multicolumn{2}{:c}{7523416} &
    \multicolumn{2}{:c|}{5530219} \\
    \hline      
  \end{tabular}

  \caption{ThreeBallot showing a vote for Bob Smith 
\label{fig:votingmodel}
  \label{fig:votingmodellatex}
}
\end{figure}

\ThreeBallot~\cite{Rivest06,Rivest07threeProtocols} is a voting protocol that strives to achieve some desirable properties, 
such as anonymity and verifiability of voting, without the use of cryptography.
We provide a brief illustration of the protocol hereafter and refer to
\cite{Rivest06,Rivest07threeProtocols} for further details. Each voter identifies herself at
the poll site, and receives a ``multi-ballot'' paper to vote with.
The multi-ballot consists of three vertical ribbons -- identical
except for ID numbers at the bottom (see Fig.~\ref{fig:votingmodel},
presented after~\cite{Rivest06}).  The voter fills in the
multi-ballot, separates the three ribbons, and casts them into the
ballot box.
The ballot box has
the property, as usual, that it effectively scrambles the
ballot order, destroying any indication of which triple of
ballots originally went together, and what order ballots
were cast in.
To vote for a candidate, one must mark exactly two (arbitrary)
bubbles on the row of the candidate.  To not vote for a candidate, one
must mark exactly one of the bubbles on the candidate's row (again,
arbitrarily).  In all the other cases the vote is invalid.  The
ballots are tallied by counting the number of bubbles marked for each
candidate, and then subtracting the number of voters from the count.

While voting, the voter also receives a copy of one of her three ballots, and she can take it home.
After the election has terminated, all the ballots are scanned and
published on a web bulletin board.  In consequence, the voter can
check if her receipt matches a ballot listed on the bulletin board.
If no ballot matches the receipt, the voter can file a complaint.

Since \ThreeBallot is not a cryptographic protocol, it does not
heavily rely on computers and counting can be done directly. Moreover,
voters have no responsibility to ensure the integrity of cryptographic
keys, and the security process in their vote is essentially the same
as with traditional ballots.

\smallskip
\noindent\textbf{Properties.}  \ThreeBallot was proposed to provide
several properties that reduce the possibility of electoral fraud.
\emph{Anonymity} (cf.~e.g.~\cite{Moran2014}) 
requires that no agent should ever know how another voter voted,
except in cases where it is inevitable, such as when all voters voted
for the same candidate.  Anonymity is important as it limits the
opportunities of coercion and vote-buying.  The latter kind of
properties is usually captured by the notions of coercion-resistance
and
receipt-freeness~\cite{benaloh1994receipt,delaune2005receipt,juels2005coercion}.
In particular, \emph{coercion-resistance} requires that the voter
cannot reveal the value of her vote beyond doubt, even if she fully
cooperates with the coercer. As a consequence, the coercer has no way
of deciding whether to execute his threat (or, dually, pay for the
vote).  A preliminary formalization of coercion-resistance and
receipt-freeness in ATL has been presented
in~\cite{Tabatabaei16expressingCR}.

Finally, \emph{end-to-end voter verifiability}~\cite{Ryan15verifiability,computer-ate-my-vote}
provides a way to verify the outcome of the election by allowing voters to audit the information published by the system.
Typically, the focus is on individual verifiability: each voter should be able check if her vote has been taken into account and has not been altered.

\subsection{iCGS Models}

Here we present three iCGS models of the \ThreeBallot voting protocol.
All models have been specified in ISPL (Interpreted System Programming
Language), the input language of the MCMAS model checker for
multi-agent systems \cite{LomuscioQuRaimondi15}. We do not model
several aspects of the voting system: the ID of each ribbon, the copy
of the ribbon which is given back to each voter after casting her
ballot, the possibility for voters to verify the presence of the
ribbon they are given back after voting.
Moreover, we model a single attacker who is also a voter and, as such,
follows the voting protocol and does not interact in any particular
way with the other agents.

In the iCGS below, each agent is represented by means of her local
variables and their evolution.  The vote collector and bulletin board
(BB) are modeled by the Environment agent (called Env).  This agent
posesses local variables modeling the fact that the voting process is
open and the values of ribbons on the BB.  These variables are
observable by all voters, including the attacker.  Env also owns
private variables for collecting ribbons and can perform the
three actions $stop$, $collect$, $nop$ 
to close the voting process, to collect the votes, and finally to loop
after the publication of the BB.

The agents representing voters have each a private variable
representing their vote for each candidate, and they share three
``ballot'' variables with the environment Env.  These variables
represent the ribbons created by the ``voting machine''.  Casting the
vote is modeled by creating the three ribbons, consistently with the
choice for each candidate. We assume that votes are immediately cast
in the initial state.  Being visible by Env, the values of the three
ribbons are copied by Env onto the (variables represented on the) BB
in a random order.
Each agent has two actions: $vote$, by which the voter casts her vote,
and $nop$, a non-voting or idle action. Action $vote$ is enabled only
in the initial state, $nope$ is always enabled.  All agent variables
are never modified during the voting process.

In the first iCGS, denoted $\G_{tot}$, for each agent choice, all
configurations of the three ribbons that are compatible with the
agent's choice may occur.  The communication between each agent and
Env is entirely at Env's charge, who has direct access to agents'
ribbons and copies them onto the BB. Also, copying is done randomly:
Env chooses a not-yet-copied ribbon from some voter who has cast her vote
(boolean variables are provided to help Env identify these situations)
and copies it onto a free position on the BB.

With the second model, denoted $\G_{lex}$, we model a voting machine
which sorts, according to the lexicographic order, the three ribbons
produced for the agent's vote, and places the greatest one in the
first ``ballot'' variable of the voter, the second greatest in the second
variable, and the least in the third variable.  Hence, for each
choice of an agent, there are still several configurations of ribbons
that are produced, but we no longer consider all permutations of a given
configuration, but only a single representative.

Finally, we modify $\G_{lex}$ into a third model, in which Env no
longer copies ribbons on the BB, but rather counts the votes for each
candidate by peeping at the "ballot" variables of each voter.  This
model is denoted $\G_{count}$.

More formally,
in the case of $\G_{tot}$ for $n$ voters and $c$ candidates, each
global state has the form $( vopen, pub, (ribb_\ell)_{1\leq \ell\leq
  3n}, (ch_i, v_i)_{1 \leq i \leq n}, (s_{ij})_{1\leq i\leq n, 1\leq j
  \leq 3} )$ such that
\begin{enumerate}
\item The Boolean variable $vopen$ is true when the vote is open, while
  $pub$ signals that all ribbons of agents that have voted have been
  copied on the BB.
\item The local states for voter $i$ is
$( vopen, pub, ribb_1, ribb_2,...ribb_{3n}, v_i, s_{i1}, s_{i2}, s_{i3})$.
This means that each voter sees the BB during the process of copying ribbons.
This is not harmful for anonymity since indistinguishability is state-based, which means 
agents do not remember their observations. 
\item Integer $1\leq ch_i\leq c$ specifies the choice of voter $i$.
\item Boolean $v_i$ ($1\leq i\leq n$) registers whether agent $i$ has
  voted.
\item Variables $s_{ij}$ ($1 \leq j \leq 3$) represent the three "strips" of the ballot for voter $i$. They
are shared between each agent and Env, who copies them onto the BB.
Value range is $\{0, 2^{c}-1\} \cup \{\bot\} $ with $s_{ij} = \bot$ denoting that
 strip $s_{ij}$ has been copied on the BB.
\item Integer variables $ribb_\ell$ ($1 \leq \ell \leq 3n $) represent the content of the BB.
Value range is $\{0, 2^{c}-1\} \cup \{\bot\} $, with $ribb_\ell = \bot$ denoting that no content of a strip
has been copied into this ribbon on the BB.
\end{enumerate}

Initial states are such that $vopen = true$, $v_i = false$ for all
$i\leq n$, variables $ribb_\ell$ are set to the \emph{undefined} value
$\bot$
and, for variables $s_{ij}$ we have the following rules modeling the
creation of a triple of ribbons compatible with the choice of a
candidate: for each voter $i$, let $b_{jk}=b^i_{jk}$ be the bit
representing the bubble on the line corresponding to candidate $k$ of
the $j$th ballot of $i$'s vote, as represented by the value of $ch_i$.
A tuple $(b_{jk})_{1\leq j \leq 3, 1 \leq k \leq c}$ is \emph{compatible} with choice $ch_i$
if the following holds:
	\begin{enumerate}
 	\item if $ch_i = k$ then for some $p \leq 3$, $b_{pk} = 0$ and
          for all $p' \neq p$, $b_{p'k} = 1$
 	\item if $ch_i \neq k$ then for some $p \leq 3$, $b_{pk} = 1$ and for all $p' \neq p$, $b_{p'k} = 0$
	\end{enumerate}

Denote $B(ch_i)$ the set of bit tuples $(b_{jk})_{1\leq j \leq 3, 1
  \leq k \leq c}$ compatible with $ch_i$.  Denote further by $R(ch_i)$
the transformation of these bit tuples into integer triples modeling
the valid ballots compatible with choice $ch_i$:
\[
R(ch_i) = \{ (st_{j})_{1 \leq j \leq 3} \mid st_{j} = \sum_{1 \leq k \leq c} b_{jk} \cdot 2^{k-1},
(b_{jk})_{1\leq j \leq 3, 1 \leq k \leq c} \in B(ch_i) \}
\]
As an example, valid triples of integers compatible with a vote for
candidate $2$, for $c = 2$, are all permutations of $(3,2,0)$ together with all
permutations of $(2,2,1)$.
Then $(s_{ij})_{1 \leq j \leq 3} \in R(ch_i)$ for every $1 \leq i\leq n, 1\leq j\leq 3$.

Finally, in $\G_{tot}$ the protocol for the \ThreeBallot system is
given as follows:
\begin{enumerate}
\item actions $vote$ and $nop$ are available to all voters when
  $vopen = true$; otherwise, only $nop$ is available.
\item $stop$ and $nop$ are available to Env when $vopen = true$.
\item $collect$ and $publish$ are available to Env when $vopen =
  false$.
\end{enumerate}

Then, transitions are of the form:
\begin{gather*}
( vopen, pub, (ribb_\ell)_{1\leq \ell\leq 3n},  (\!ch_i, v_i\!)_{1 \leq i \leq n}, (\!s_{ij}\!)_{1\leq i\leq n, 1\leq j \leq 3} \!)
\xrightarrow{(a_e, a_1, a_2,..., a_n)}\\
(vopen'\!,pub'\!, (ribb'_\ell)_{1\leq \ell\leq 3n}, (ch_i', v_i')_{1 \leq i \leq n}, (s_{ij}')_{1\leq i\leq n, 1\leq j \leq 3} )
\end{gather*}
where
\begin{enumerate}
\item $vopen' = false$ if $a_e = stop$ or $vopen = false$; $vopen' = true$ otherwise.
\item For $a_i = vote$, $v_i'= true$, and for $a_i = nop$, $v'_i = v_i$.
\item For $a_e = collect$ (and hence $a_i = nop$ for all $i$)
we have the following:
\begin{enumerate}
\item there exists some subset of pairs $A \subseteq \{1,\ldots,n\}
  \times \{1,2,3\}$ and a pair $(i_0,j_0) \not \in A$ such that
\begin{enumerate}
\item $s_{ij}'= s_{ij} = \bot$ for all $(i,j) \in A$
\item $s_{i_0,j_0}' = \bot$, $s_{i_0,j_0} \neq \bot$ 
\item $s_{ij}' =s_{ij}\neq \bot$ for all $(i,j) \not \in A \cup\{(i_0,j_0)\}$;
\end{enumerate}
\item there exists some $B \subseteq \{1,\ldots,3n\}$ with $ card(B) =
  card(A)$ and some integer $k\not \in B$, $1\leq k\leq 3n$ such that
\begin{enumerate}
\item $ribb'_\ell = ribb_\ell \neq \bot$ for all $\ell \in B$
\item $ribb_k = \bot$, $ribb'_k = s_{i_0,j_0}$
\item $ribb_\ell = ribb'_\ell = \bot$ for all $\ell \not \in B\cup \{k\}$.
\end{enumerate}
\end{enumerate}
\item Action $a_e = publish$ can only be executed when, for each $i$,
  either $s_{i1} = s_{i2} = s_{i3} = \bot$ or $v_i=false$, and its
  effect is to modify only $pub'= true$, all other variables remaining
  unchanged.
\end{enumerate}

In the iCGS $\G_{lex}$ transitions are identical to the above, the
only difference being in the initial states, more specifically in the
configuration of variables $s_{ij}$.  Intuitively, for each choice
$ch_i$ of voter $i$, we only keep those initial configurations in
which $(s_{i1},s_{i2},s_{i3})$ is the maximal, in lexicographic order,
among the encodings $R_{ch_i}$ of $ch_i$.  Formally, given a triple of
integers $t = (t_1,t_2,t_3)$, we denote with $Perm(t)$ the set of all
permutations of triple $t$.
Then, the initial states in $\G_{lex}$ are the initial states
of $\G_{tot}$ for which $(s_{i1},s_{i2},s_{i3}) = \max Perm\big(
(s_{i1},s_{i2},s_{i3})\big )$,
the maximum being considered under the lexicographic order.

Finally, the iCGS $\G_{count}$ is similar to $\G_{lex}$ but all
variables $ribb_\ell$ are replaced with $c$ variables $(co_k)_{1\leq
  k\leq c}$.  The local states for agent $i$ are then of the form $(
vopen, pub=0, v_i, s_{i1}, s_{i2}, s_{i3} )$ and $( vopen, pub=1,
co_1,\ldots,co_{c}, v_i, s_{i1}, s_{i2}, s_{i3} )$.
The specification of transitions is then the same as for $\G_{tot}$,
except
for $a_e = collect$ and items 3.(b.i)-3.(b.iii) above (defining the
updates of variables $ribb_\ell$), which are replaced by the
following:
\begin{itemize}
\item[3.(b')] For each $1\leq k \leq c$, $co'_k = co_k +
  d_{i_0j_0k}$, where $d_{i_0j_0k}$ is the $k$-th least significant
  bit of $s_{i_0j_0}$.
\end{itemize}

\subsection{Coercion freeness and anonymity properties for the \ThreeBallot}

In this section we present the formulas that are of interest for the
verification of \ThreeBallot.  We verify two types of formulas: a
variant of coercion resistance \cite{Tabatabaei16expressingCR} and a
variant of anonymity.  The coerction resistance property specifies the
fact that the attacker $att$ has no strategy by which he could know
how agent $i$ has voted ($i \neq att$):
\[
\varphi_i = \llangle att \rrangle F \big((pub \wedge  v_i) \rightarrow \bigvee_{1\leq j\leq nc} K_{att} (j=ch_i)\big)
\]

Recall that, in our model the attacker is also a voter, which
corresponds to situations in which a voter fully cooperates with the
attacker.
Additionnaly, as already noted in Section 2, the knowledge operator is
definable in ATL with the subjective interpretation as $K_{att} \varphi =
\llangle att \rrangle \varphi U \varphi$.

The anonymity property does not require ATL but rather CTLK, the
combination of branching-time temporal logic with epistemic logic of
e.g. \cite{LomuscioRaimondi06b}.  The $AG$ operator from CTLK, whose meaning is that
along all paths the target subformula holds globally, can be defined
as usual in ATL as follows:
\[
AG \varphi \equiv \llangle \emptyset \rrangle G \varphi
\]

Then, the anonymity property of interest for \ThreeBallot{} is that
any of the agents cannot know, at any time instant, the way another
agent has voted:
\[
\varphi_i^{c} = AG ( \bigwedge_{1 \leq j \leq nc} \neg K_{att} p_{ch_i = j} )
\]
Note that anonymity is not ensured when there are at most three
voters, including the attacker (who may vote).

\subsection{Bisimulations for $\G_{tot}$, $\G_{lex}$ and $\G_{count}$}

The three iCGS defined in the previous section appear to be naturally
related, in particular w.r.t.~the properties pertaining to the
attacker modifying the outcome of the vote or breaking anonymity.  The
interest in simplifying the model is that checking the coercion
resistance property can be done faster and with less memory on
$\G_{count}$ than on $\G_{lex}$, which, on its turn, requires less
time and memory than $\G_{tot}$, as we will see in the last section on
experimental results.  In this section we show that the three models
are bisimilar for the attacker, for the set of atomic propositions
that refer only to choices of agents.  Bisimulations formalize the
``natural relation'' between these iCGS and allows us to check
coercion resistance on the smallest iCGS and then transfer the result
to the two others, in particular to the original model $\G_{tot}$.
Note that these bisimulations work because the properties do not refer
to the status of the BB.  For instance, these bisimulations would not
be useful for simplifying systems for verifiability \cite{DKR09}.

Formally, for each choice for an agent $i$ to vote for candidate $j$,
we introduce an atomic proposition $p_{ch_i=j}$, which holds true only
in those states in which $ch_i = j$.  Then, if we denote the attacker
$att = n$ and $AP = \{p_{ch_i=j} \mid 1\leq i\leq n, 1\leq j\leq
c \}$, we prove that the following relation is an
$\{att\}$-bisimulation over $AP$ between $\G_{tot}$ and $\G_{lex}$: 
\begin{eqnarray*}
& ( vopen,pub, (ribb_\ell)_{1\leq \ell\leq 3n}, (ch_i, v_i)_{1 \leq
i \leq n},  (s_{ij},a_{ij})_{1\leq i\leq n, 1\leq j \leq 3} ) \altbisim{\{att\}}^1 & \\ &
( vopen', pub',(ribb'_\ell)_{1\leq \ell\leq 3n},  (ch_i', v_i')_{1 \leq i \leq n}, (s_{ij}',a_{ij}')_{1\leq i\leq n, 1\leq j \leq 3} ) &
\end{eqnarray*}
iff the following holds:
\begin{enumerate}
\item $vopen = vopen'$, $pub=pub'$, $v_i = v_i'$, $ch_i = ch_i'$ for all $1\leq i\leq n$, 
$s_{att,j} = s_{att,j}'$ and $a_{att,j} = a_{att,j}'$ for all $1\leq j\leq 3$ 
and $ribb_\ell = ribb'_\ell$ for all $1\leq \ell \leq 3n$.
\item For every $1\leq i\leq n$, if we denote 
$b_{jk}$ the $k$th least significant bit of $s_{ij}$ and 
$b'_{jk}$ the $k$th bit of $s'_{ij}$,
then both $(b_{jk})_{1\leq j\leq 3, 1\leq k\leq c},(b'_{jk})_{1\leq j\leq 3, 1\leq k\leq c} \in B(ch_i)$.

\item Denote $\rho_i$ the $S_3$-permutation of $(s_{i1},s_{i2},s_{i3})$ into $(s'_{i1}, s'_{i2},s'_{i3})$, 
i.e. $s_{ij} = s'_{i\rho_i(j)}$. Also when $s_{ij} = s'_{ij} = \bot$ we put $\rho_i = id_{\{1,2,3\}}$.
Then $a_{ij} = a'_{i\rho_i(j)}$ for all $1\leq i\leq n-1,1\leq j\leq 3$.
\end{enumerate}
Intuitively, item (3) above says that $(b'_{jk})$ is the largest, in
lexicographic order, among all tuples in $B_{ch_i}$ that are
permutations of $(b_{jk})$.

For $\G_{lex}$ and $\G_{count}$, we consider the following
$\{att\}$-bisimulation over $AP$: 
\begin{eqnarray*}
& ( vopen, pub,
(\!ribb_\ell\!)_{1\leq \ell\leq 3n},\! (\!ch_i, v_i\!)_{1 \leq i \leq
n}, (s_{ij},a_{ij})_{1\leq i\leq n, 1\leq j \leq 3} ) \altbisim{\{att\}}^2 &\\
&( vopen', pub', (co_k)_{1\leq k\leq c}, (ch_i', v_i')_{1 \leq i \leq n}, (s_{ij}',a_{ij}')_{1\leq i\leq n, 1\leq j \leq 3} )&
\end{eqnarray*}
where:
\begin{enumerate}
\item $vopen = vopen'$, $pub=pub'$, $v_i = v_i'$, $ch_i = ch_i'$, 
$s_{ij} = s_{ij}'$ and $a_{ij} = a_{ij}'$ for all $1\leq i\leq n, 1\leq j\leq 3$.
\item For every $1\leq \ell \leq 3n$ and $1 \leq k\leq c$, if we denote $b_{\ell k}$ 
the $k$th least significant bit on the ribbon $ribb_\ell$, then:
\[
co_k = \sum \big\{ b_{\ell k} \mid ribb_\ell \neq \bot, 1\leq \ell \leq 3n \big\} 
\]
\end{enumerate}

To prove that these relations are indeed alternating bisimulations,
note that the condition 1.(a) is immediately satisfied as whenever
$q \altbisim{\{att\}}^{\iota} q'$, for $\iota = 1,2$, we must have that
$(ch_i = j) \in q$ iff $(ch_i = j) \in q'$.

To prove properties 1.(b) and its dual for $\altbisim{\{att\}}^1$,
consider states $q,r$ in $\G_{tot}$ and $r \in \G_{lex}$ such that
$q \altbisim{\{att\}}^1 q'$ and  $q \sim_{att} r$.
Then it is the case that
\begin{align*}
q & = ( \!vopen, \!pub, \!(\!ribb_\ell\!)_{1\leq \ell\leq 3n}, \!(\!ch_i, v_i\!)_{1 \leq i \leq n}, \!(\!s_{ij},\!a_{ij}\!)_{1\leq i\leq n, 1\leq j \leq 3} )\\
q'& = (\! vopen, \!pub, \!(\!ribb_\ell\!)_{1\leq \ell\leq 3n}, \!(\!ch_i, v_i\!)_{1 \leq i \leq n}, \!(s_{ij}'\!,\!a_{ij}'\!)_{1\leq i\leq n, 1\leq j \leq 3} )\\
r & = ( \!vopen, \!pub, \!(\!ribb_\ell\!)_{1\leq \ell\leq 3n}, \!(\!\br {ch}_i, \br v_i\!)_{1 \leq i \leq n}, \!(\!\br s_{ij},\!\br a_{ij}\!)_{1\leq i\leq n, 1\leq j \leq 3} )
\end{align*}
with $q,q'$ related as per the definition of $\altbisim{\{att\}}^1$
above and $\br {ch}_{att}= ch_{att}$, $\br s_{att,j} = s_{att,j}'$ and
$\br a_{att,j} = a_{att,j}'$ for all $1\leq j\leq 3$. Then set
{\footnotesize
\begin{align*}
r' & = ( vopen,pub, (ribb_\ell)_{1\leq \ell\leq 3n}, (ch_i', v_i')_{1 \leq i \leq n-1}, (\br {ch}_{att}, \br v_{att}), 
 (s_{ij}',a_{ij}')_{1\leq i\leq n-1, 1\leq j \leq 3}, (\br s_{att,j},\br a_{att,j})_{1\leq j\leq 3} )
\end{align*}
} and we obtain the desired result, that is,
$q'\altbisim{\{att\}}^1r'$ and $r \sim_{att}r'$.  The mirror argument
works as well: given $q \altbisim{\{att\}}^1 q'$ and $q' \sim_{att}
r'$, we can find $r$ such that $q\altbisim{\{att\}}^1r$ and $r \sim_{att}r'$.

Conditions 1.(b) and its dual for $\altbisim{\{att\}}^2$ can be proved
similarly, by observing that, given $q, r \in \G_{lex}$, $q'\in \G_{count}$, such that $q \altbisim{\{att\}}^2 q'$ and
$q \sim_{att} r$, then it is the case that 
\begin{align*}
q & = (\! vopen,\!pub, \! (\!ribb_\ell\!)_{1\leq \ell\leq 3n}, \!(\!ch_i, \!v_i\!)_{1 \leq i \leq n}, \!(\!s_{ij},\!a_{ij}\!)_{1\leq i\leq n, 1\leq j \leq 3} )\\
q'\!& = ( \!vopen, \!pub,\! (co_k)_{1\leq k \leq c}, \!(\!ch_i,\! v_i\!)_{1 \leq i \leq n}, \!(s_{ij}'\!,\!a_{ij}'\!)_{1\leq i\leq n, 1\leq j \leq 3} ) \\
r & = ( \!vopen, \!pub, \!(ribb_\ell)_{1\leq \ell\leq 3n}, \!(\br {ch}_i, \!\br v_i\!)_{1 \leq i \leq n}, \!(\!\br s_{ij},\!\br a_{ij}\!)_{1\leq i\leq n, 1\leq j \leq 3} )
\end{align*}
with the same relation between variables of $q$ and $r$ as
above. Then set
{\footnotesize
\begin{align*}
r' & = ( vopen, pub, (co_k)_{1\leq k \leq c}, (ch_i', v_i')_{1 \leq i \leq n-1}, (\br {ch}_{att}, \br v_{att}), 
 (s_{ij}',a_{ij}')_{1\leq i\leq n-1, 1\leq j \leq 3}, (\br s_{att,j},\br a_{att,j})_{1\leq j\leq 3} )
\end{align*}
}
and we obtain that $q'\altbisim{\{att\}}^2r'$ and $r \sim_{att}r'$.

Further, for conditions 1.(c) and its dual, notice first that
for any state $q\in \G_{tot}$ or $q \in \G_{lex}$, $C_{att}(q)$ is the
equivalence class of $q$ w.r.t.~$\sim_{att}$, that is, if $q = (
vopen, pub, (ribb_\ell)_{1\leq \ell\leq 3n}, (ch_i, v_i)_{1 \leq
i \leq n}, (s_{ij},a_{ij})_{1\leq i\leq n, 1\leq j \leq 3} )$, then
$C_{att}(q)$ includes all and only states where the local state for
$att$ is of the form $((ribb_\ell)_{1\leq \ell \leq 3n}, ch_{att},
v_{att}, (s_{att,j}, a_{att,j})_{1\leq j\leq 3})$.  Similarly, for
$q'\in \G_{count}$ with $q' \!= ( \!vopen, pub, (\!co'_k\!)_{1\leq
k\leq c}, (\!ch_i'\!,\! v_i'\!)_{1 \leq i \leq n},
(\!s_{ij}'\!,\!a_{ij}'\!)_{1\leq i\leq n, 1\leq j \leq 3} )$,
$C_{att}(q')$ includes all and only states where the local state for
$att$ is of the form $( (co_k)_{1\leq k\leq c}, ch_{att}, v_{att}, \allowbreak
(s_{att,j}, a_{att,j})_{1\leq j\leq 3})$.

Then, in all three iCGS, on each neighbourhood $C_{att}(q)$, only one
or two partial strategies for $att$ can be defined, depending on
whether the vote is open or not.
Therefore, we define the mapping
$ST_{C^{\G_{tot}}_{att}(\cdot),C^{\G_{lex}}_{att}(\cdot)}$ to
associate to each partial strategy prescribing $nop$ to $att$ in some
$C^{\G_{tot}}_{att}(q)$ the quasi-identical strategy prescribing the
same action in $C^{\G_{lex}}_{att}(q')$. Similarly, to the partial
strategy prescribing $vote$ to $att$ in $C^{\G_{tot}}_{att}(q)$,
$ST_{C^{\G_{tot}}_{att}(\cdot),C^{\G_{lex}}_{att}(\cdot)}$ associates
the strategy prescribing $vote$ in $C^{\G_{lex}}_{att}(q')$.  The dual
mapping $ST'$ is defined similarly, and analogous definitions work for
bisimulation $\altbisim{\{att\}}^2$.  Notice also that these definitions
satisfy the constraints on strategy simulators.

To prove condition 1.(c), consider first the strategy $vote_{att}$
prescribing $vote$ for $att$ on $C^{\G_{tot}}_{att}(q)$ and consider
some state $r'$ such that $q' \xrightarrow{vote_{att}} r'$.  Since when
voting is enabled, Env does not collect votes, $r'$ has the same BB as
$q'$ and all booleans variables $a_{ij}$ are false.  Therefore, we can
choose a state $r$ that has the same values as $r'$ for the local
variables of all voters and the same BB as $q$. Then, we obtain 
$q'\xrightarrow{ST(vote_{att})} r'$ and $r \altbisim{\{att\}}^1r'$.  A
similar line of reasoning works for $\altbisim{\{att\}}^2$ as well.
The same argument applies if the strategy for the attacker is $nop_{att}$
and $q$ and $q'$ are states in which voting is open.

Consider now a state $q$ in which voting is closed 
and the strategy $none_{att}$ prescribing action $nop$ for $att$ on $C^{\G_{tot}}_{att}(q)$
and consider again some state $r'$ such that 
$q' \xrightarrow{nop_{att}} r'$.
Then the only agent which executes a non-idle action on the above
transition is Env, who copies one
of the ribbons onto the BB. This transition can then be simulated in
$\G_{tot}$ by copying the same ribbon (but which might be stored at a
different position in $q$ than in $q'$) onto the BB.  A mirror
argument can be used to prove 1.(c) for $ST'$.

Thirdly, for proving condition 1.(c) for $q\altbisim{\{att\}}^2q'$,
note that the same considerations above apply for the case of a
transition from a state $q'$ in which the voting is open.  For the
case of states $q$, $q'$ where the voting is closed and hence only
strategy $nop_{att}$ is available to the attacker, we note that the
only action that is compatible with $q \xrightarrow{nop_{att}}r$ in
both models is Env collecting votes.  This corresponds in
$\G_{count}$ to an action in which Env counts votes, and
hence we can find a state $r$ that is an $nop_{att}$-compatible
successor of $q$, has the same local states for voters as $r'$,
and in which each counter $co_k$ 
keeps the sum of the bullets
on $k$th line on the copied ribbons from $r$.  This will ensure that
$r \altbisim{\{att\}}^2 r'$.

Finally, as regards condition (2), if $q_1 \altbisim{\{att\}}^1 q'$
and $q_2 \altbisim{\{att\}}^1 q'$, then the attacker must have made
the same choice in $q_1$ and $q_2$ (possibly none), and her "ballot stripes" $s_{att,j}$ and "bookkeeping bits" $a_{att,j}$
must be the same in $q_1$ and $q_2$.
But this entails the indistinguishability of $q_1$ and $q_2$ for the attacker, i.e., then $C_{att}(q_1) =
C_{att}(q_2)$.  The proof of (2) for $\altbisim{\{att\}}^2$ is similar.

To conclude, all iCGS $\G_{tot}$, $\G_{lex}$, and $\G_{count}$ are
bisimilar for the attacker. We will make use of this fact in the next
section on experimental results.

\begin{remark}
Note that the two relations $\altbisim{\{att\}}^1$ and
$\altbisim{\{att\}}^2$ are also bisimulations on models
$\G_{tot}, \G_{lex}$ and $\G_{count}$ in the sense of (the symmetric
variant of) Def.~5.1 in \cite{CohenDLR09}.  Therefore, applying
\cite[Lemma 5.2]{CohenDLR09} both ways, any CTLK formula that is satisfied in any of the
models, is satisfied in all of them. This justifies our use of CTLK
operators in specification $\varphi^c_i$ above.

On the other hand, note that CTLK bisimulations in the sense
of \cite{CohenDLR09} do not preserve ATL formulas under imperfect
information.  Figure~\ref{fig-timed-epistemic-counterexample} provides
a counterexample, with two 1-agent models that are timed epistemic
bisimilar, but do not satisfy the same formulas in ATL.

\end{remark}

\newcommand{\mirrormodels}{\mathrel{\reflectbox{\rotatebox[origin=c]{0}{$\models$}}}}

\definecolor{acsi-blue}{RGB}{44,130,202}
\definecolor{acsi-green}{RGB}{128,197,53}
\definecolor{acsi-orange}{RGB}{245,126,27}
\definecolor{postergreen}{HTML}{008000}
\definecolor{posterpurple}{HTML}{800080}

\begin{figure}[ht]
\begin{center}
\begin{tikzpicture}[auto,node distance=2.5cm,->,>=stealth',shorten
>=1pt,semithick]
\tikzstyle{every state}=[fill=white,draw=black,text=black,minimum size=18pt,circle]
\tikzstyle{every initial by arrow}=[initial text=]
\node[state,draw=red,label = {right:{$\not \models \llangle alice \rrangle F p $}}] 	(s0)   	{};
\node[state,label={left:$s_1$}]    (s1) 	[below left of=s0]	{};
\node[state,label={right:$s_2$}]    (s2) 	[below right of=s0]	{};
\node[state,label={left:$s_3$}]    (s3) 	[below of=s1]	{$p$};
\node[state,label={right:$s_4$}]    (s4) 	[below of=s2]	{};
\node[state,draw=red,label = {right:{$ \models\llangle alice \rrangle F p$}}] 	(t0)   	[right of=s0,node distance = 6.5cm]{};
\node[state,label={left:$t_1$}]    (t1) 	[below left of=t0]	{};
\node[state,label={right:$t_2$}]    (t2) 	[below right of=t0]	{};
\node[state,label={left:$t_3$}]    (t3) 	[below of=t1]	{$p$};
\node[state,label={right:$t_4$}]    (t4) 	[below of=t2]	{};
\draw 	(s0) 	edge[draw=red] node[left] {\textcolor{red}{$*$}} (s1);
\draw 	(s0) 	edge[draw=red] node[left] {\textcolor{red}{$*$}} (s2);
\draw 	(s1) 	edge[draw=blue] node[right] {\textcolor{blue}{$a$}} (s3);
\draw 	(s1) 	edge node[above left=0.2] {{$b$}} (s4);
\draw 	(s2) 	edge[draw=blue] node[right] {\textcolor{blue}{$a$}} (s4);
\draw 	(s2) 	edge node[above right=0.2] {{$b$}} (s3);
\draw   (s1) edge[dashed,-] node[above]  {$\textcolor{posterpurple}{alice}$} (s2);
\draw 	(t0) 	edge[draw=red] node[left] {\textcolor{red}{$*$}} (t1);
\draw 	(t0) 	edge[draw=red] node[left] {\textcolor{red}{$*$}} (t2);
\draw 	(t1) 	edge[draw=blue] node[right] {\textcolor{blue}{$a$}} (t3);
\draw 	(t1) 	edge node[above left=0.2] {{$b$}} (t4);
\draw 	(t2) 	edge node[right] {{$b$}} (t4);
\draw 	(t2) 	edge[draw=blue] node[above right=0.2] {{\textcolor{blue}{$a$}}} (t3);
\draw   (t1) edge[dashed,-] node[above]  {$\textcolor{posterpurple}{alice}$} (t2);
\end{tikzpicture}
\end{center}
\caption{\label{fig-timed-epistemic-counterexample}
Two models that are (timed epistemic) bisimilar in the sense
of \cite{CohenDLR09}, but do not satisfy the same formula in ATL.}
\end{figure}
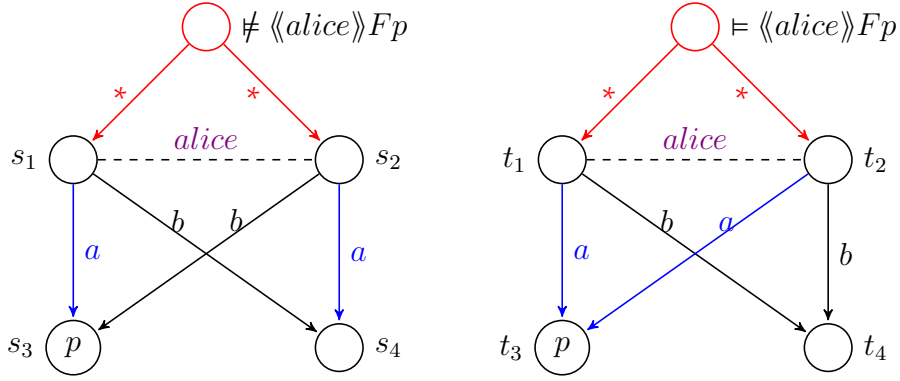

\section{Experimental Results} \label{exp}

In this section, we exhibit the improvements in running time when
checking the same properties over the three bisimilar models. The
three models are checked with growing number of voters and candidates.
For our experiments, we have used the latest version of MCMAS
(1.3.0) \cite{LomuscioQuRaimondi15}.  Tests were made on a virtual
machine running \textsf{Ubuntu 16.04.1 LTS} on a Dell PowerEdge R720
server with two Intel Xeon E5-2650 8 core processors at 2GHz, and 128
GB of RAM. 
In order to investigate certain non-intuitive results 
reported in Tables~\ref{table-atlk-2},~\ref{tab:tot}, and \ref{tab:lex},
we re-ran the experiments on a different machine with a similar 
setup of 16-core Xeon E5620 2.4GHz processor with 64GB RAM.
The purported anomalies have shown to be stable over both the tests.
In all the following tables, NA means a 2 hours timeout has been reached without obtaining any result.
The C programs that were used to generate the ISPL files for the three \ThreeBallot{} models, together with the generated 
ISPL files, are publicly available as \cite{three-ballot-repository}.

MCMAS provides two options, \texttt{-atlk 2} or \texttt{-uniform}, for
checking ATL formulas with uniform strategies, with some differences
in the semantics of ATL formulas (\texttt{-uniform} is similar
with \emph{``irrevocable strategies''}
of \cite{Agotnes07irrevocable}).
We observed that neither of these
options were stable, and lead to a number of experiments ending with
inconsistent results or MCMAS terminating abnormally (segfault on
null-pointer assignment in one of the fixpoint computations related
with ATL satisfiability).
We refer the interested reader to~\cite{HoekLomuscioWooldridge06,BusardPQR15}.

Table~\ref{table-atlk-2} reports the running times for a selection
of $\G_{tot}$ and $\G_{lex}$ models. Note that several surprising data points are present.
Namely, in the case of $(2v,3c)$ the running time is smaller for $\G_{tot}$
than for $\G_{lex}$ despite considerably larger state space (but smaller BDD).
Moreover in case of $(2v,4c)$ the computations for $\G_{lex}$ timed out
at the stage of computing the set of reachable states 
while the verification of $\G_{tot}$ completed successfully in an hour.
Model checking the cases of $(3v,2c)$ and $(4v,2c)$ resulted in segmentation fault
for both $\G_{tot}$ and $\G_{lex}$ on both the machines.
Model checking the cases of $(3v,3c)$ and $(4v,3c)$
did not succeed in reasonable time for both models.

\begin{table}[t]
\footnotesize
\centering
\begin{tabular}{c | c | c | c | }
	  \cline{2-4}
          
	  & \multicolumn{3}{c|}{(\# voters, \# candidates)} \\
	  \cline{2-4}
	  & (2v,2c) & (2v,3c) & (2v,4c)\\
  	  \cline{2-4} \cline{1-4}
\multicolumn{1}{ |r||  }{\multirow{3}{*}{$G_{tot}$} } & 1.5 s & 12.8 s & 3714.0 s  \\
  	  \multicolumn{1}{ |r||  }{ } &   $|S| \approx 5.9e+06$ & $|S| \approx 3.3e+08$ & $|S| \approx 8.8e+09$ \\
  	  \multicolumn{1}{ |r||  }{ } &  $|BDD| \approx 1.4e+07$ & $|BDD| \approx 5.5e+07$ & $|BDD| \approx 4.0e+08$  \\
  	
  	  \hline
\multicolumn{1}{ |r||  }{\multirow{3}{*}{$G_{lex}$} } & 1.3 s & 29.0 s & {\multirow{3}{*}{NA} } \\
  	  \multicolumn{1}{ |r||  }{ } &  $|S| \approx 2.9e+05$ & $|S| \approx 1.1e+07$ & \\
  	  \multicolumn{1}{ |r||  }{ } &  $|BDD| \approx 1.5e+07$ & $|BDD| \approx 6.1e+07$ & \\
  	  \hline
 
\hline
\end{tabular}
\caption{\small MCMAS statistics for coercion freeness, $\G_{tot}$ and $\G_{lex}$ with the \texttt{-atlk 2} flag}
\label{table-atlk-2}
\end{table}

To address the segfault issues,
we also checked the coercion resistance property with \texttt{-atlk 0}
option, which utilizes ATL with perfect information.  This is
nevertheless consistent with our theoretical setting since all tests for more than 3 candidates
show that the formulas are false, and whenever a positive ATL formula
is false under the perfect information semantics, it is also false
under the imperfect information semantics, and hence preserved by
alternating bisimulations.
Tables~\ref{tab:tot} and~\ref{tab:lex} show the details of the
verification of the only configurations for which MCMAS produces
results in reasonable time (timeout = 2 hours) for, respectively,
$\G_{tot}$ and $\G_{lex}$.  Again, while the state-space of $\G_{lex}$
is consistently (and predictably) smaller than the state-space of
$\G_{tot}$, the size of corresponding BDDs can increase. This is
notable in particular in the case of $(2v,3c)$ where model checking is
twice faster for $\G_{tot}$ than for $\G_{lex}$.  The case of
$(2v,4c)$ that terminated successfully for $\G_{tot}$ and timed-out
for $\G_{lex}$ (the entry omitted from the table) is similar.

Our interpretation for the anomalies between state space, BDD size and
running time is that the "ribbon" variables being more constrained in
$\G_{lex}$ than in $\G_{tot}$, the default variable ordering produces
larger BDDs for the more constrained model.  On the other hand, the
anomaly being present for both the \texttt{-atlk 0} and
the \texttt{-atlk 2} flags, it does not seem to come from the
difference in dealing with uniform strategies than with strategies
with perfect information.  Note that the same anomalies are present
for the CTLK model-checking instances below.

\begin{table}[th]
\footnotesize
\centering
\begin{tabular}{r | c || c | c | c |}
	  \cline{3-5}
	  \multicolumn{2}{c|}{} & \multicolumn{3}{c|}{\# candidates} \\
	  \cline{3-5}
	  \multicolumn{2}{c|}{} & 2c & 3c & 4c \\
  	  \cline{3-5} \cline{1-5}
  	  \multicolumn{1}{ |r|  }{\multirow{4}{*}{\rotatebox{90}{\# voters }} } & \multicolumn{1}{ |r||  }{\multirow{2}{*}{2v} } & 1.3 s & 13.3 s  & 3558.5 s \\
  	  \multicolumn{1}{ |r|  }{ } &  &  $|S| \approx 5.9e+06$ & $|S| \approx 3.3e+08$ & $|S| \approx 8.8e+09$
\\
  	  \multicolumn{1}{ |r|  }{ } &  &  $|BDD| \approx 1.2e+07$ & $|BDD| \approx 4.5e+07$ & $|BDD| \approx 2.8e+08$\\          
  	
  	  \cline{2-5}
  	  \multicolumn{1}{ |r|  }{ } & \multicolumn{1}{ |r||  }{\multirow{2}{*}{3v} } & 25.1 s&\multicolumn{1}{ |c|  }{\multirow{3}{*}{NA} } & \multicolumn{1}{ |c|  }{\multirow{3}{*}{NA} } \\
  	  \multicolumn{1}{ |r|  }{ } &  &  $|S| \approx 2.9e+10$ &  & \\
  	  \multicolumn{1}{ |r|  }{ } &  &  $|BDD| \approx 5.9e+07$ &  & \\                    

  	  \cline{2-5}
  	  \multicolumn{1}{ |r|  }{ } & \multicolumn{1}{ |r||  }{\multirow{2}{*}{4v} } & 1291.8 s&\multicolumn{1}{ |c|  }{\multirow{3}{*}{NA} } & \multicolumn{1}{ |c|  }{\multirow{3}{*}{NA} } \\
  	  \multicolumn{1}{ |r|  }{ } &  &  $|S| \approx 1.7e+14$ &  & \\
  	  \multicolumn{1}{ |r|  }{ } &  &  $|BDD| \approx 3.0e+08$ &  & \\                    

  	  \hline
\end{tabular}
\caption{\small MCMAS statistics for coercion freeness, $\G_{tot}$ and \texttt{-atlk 0} flag}
\label{tab:tot}
\end{table}

\begin{table}[th]
\footnotesize
\centering
\begin{tabular}{r | c || c | c |}
	  \cline{3-4}
	  \multicolumn{2}{c|}{} & \multicolumn{2}{c|}{\# candidates} \\
	  \cline{3-4}
	  \multicolumn{2}{c|}{} & 2c & 3c \\
  	  \cline{3-4} \cline{1-4}
  	  \multicolumn{1}{ |r|  }{\multirow{6}{*}{\rotatebox{90}{\# voters}} } & \multicolumn{1}{ |r||  }{\multirow{2}{*}{2v} } & 1.1 s & 26.4 s  \\
  	  \multicolumn{1}{ |r|  }{ } &  &  $|S| \approx 2.9e+05$ & $|S| \approx 1.1e+07$  \\
  	  \multicolumn{1}{ |r|  }{ } &  & $|BDD| \approx 1.5e+07$ & $|BDD| \approx 5.4e+07$ \\
  	  
  	  \cline{2-4}
  	  \multicolumn{1}{ |r|  }{ } & \multicolumn{1}{ |r||  }{\multirow{2}{*}{3v} } & 10.4 s &  \\
  	  \multicolumn{1}{ |c|  }{ } &  &   $|S|\approx 3.1e+08$  & NA \\
  	  \multicolumn{1}{ |r|  }{ } &  &   $|BDD| \approx 5.7e+07$ &  \\
          
  	  \cline{2-4}
  	  \multicolumn{1}{ |r|  }{ } & \multicolumn{1}{ |r||  }{\multirow{2}{*}{4v} } & 297.9 s & \multicolumn{1}{ |c|  }{\multirow{3}{*}{NA} }  \\
  	  \multicolumn{1}{ |c|  }{ } &  &  $|S| \approx 3.9e+11$  & \\
  	  \multicolumn{1}{ |r|  }{ } &  & $|BDD| \approx 8.9e+07$ &   \\          
  	  \hline
\end{tabular}
\caption{\small MCMAS statistics for coercion freeness, $\G_{lex}$ and \texttt{-atlk 0} flag}
\label{tab:lex}
\end{table}

The models $\G_{count}$ can be verified much faster with the \texttt{-atlk 2}, the number of reachable states and the BDD size decreasing
substantially.
Statistics are given in
Table~\ref{tab:count}.  

\begin{landscape}
  \begin{table}[t]
\tiny
\centering
\begin{tabular}{r | c || c | c | c | c | c | c | c |}
	  \cline{3-8}
\multicolumn{2}{c|}{} & \multicolumn{7}{c|}{\# candidates} \\
	  \cline{3-8}
\multicolumn{2}{c|}{} & 2c & 3c & 4c & 5c & 6c & 7c 
\\
  	  \cline{3-8} \cline{1-6}
\multicolumn{1}{ |r|  }{\multirow{18}{*}{\rotatebox{90}{\# voters}} } & \multicolumn{1}{ |r||  }{\multirow{2}{*}{2v} } & 0.1 s & 0.4 s & 1.5 s & 7.1 s & 94.2 s & 
  	       \multicolumn{1}{ |c|  }{\multirow{3}{*}{NA} } 
\\
  	  \multicolumn{1}{ |r|  }{ } &  &  $|S| = 6.3e+03$ & $|S|= 1.7e+05$ & $|S| = 4.6e+06$ & $|S| = 1.2e+08$  & $|S| = 2.9e+09$ & 
\\
  	  \multicolumn{1}{ |r|  }{ } & &  $|BDD| = 1.0e+07  $ & $|BDD| = 1.2e+07$ & $|BDD| = 1.7e+07$ & $|BDD| =3.9e+07$ &  $|BDD| = 9.4e+07$ & 
\\
  	
  	  \cline{2-8}
\multicolumn{1}{ |r|  }{ } & \multicolumn{1}{ |r||  }{\multirow{2}{*}{3v} } & 0.5 s & 4.0 s & 110.0 s & 2917.8 s & 
\multicolumn{1}{ |c|  }{\multirow{3}{*}{NA} } 
\\
  	  \multicolumn{1}{ |c|  }{ } &  &  $|S| = 9.3e+04$ & $|S| = 6.2e+06$ & $|S| = 4.2e+08$ & $|S| = 3.2e+10$ & 
\\
  	  \multicolumn{1}{ |r|  }{ } & &  $|BDD| = 1.3e+07$ & $|BDD| = 2.9e+07$ & $|BDD| = 6.9e+07$ & $|BDD| = 3.1e+08$ & 
\\

\cline{2-7}
  	  \multicolumn{1}{ |r|  }{ } & \multicolumn{1}{ |r||  }{\multirow{2}{*}{4v} } & 1.3 s & 32.7 s & 762.5 s & 
\multicolumn{1}{ |c|  }{\multirow{3}{*}{NA} } 
  	  \\
  	  \multicolumn{1}{ |c|  }{ } &  &  $|S| = 7.8e+06$ & $|S| = 3.6e+08$ & $|S| = 1.6+12$ & 
\\
  	  \multicolumn{1}{ |r|  }{ } & &  $|BDD| = 1.6e+07$ & $|BDD| = 5.6e+07$ & $|BDD| = 1.7e+08$ & 
\\

  	  \cline{2-6}
\multicolumn{1}{ |r|  }{ } & \multicolumn{1}{ |r||  }{\multirow{2}{*}{5v} } & 5.7 s & 341.2 s & 
\multicolumn{1}{ |c|  }{\multirow{3}{*}{NA} }  
  	  \\
  	  \multicolumn{1}{ |c|  }{ } &  &  $|S| = 1.7e+08$ & $|S| = 2.3e+11$ & 
\\
  	  \multicolumn{1}{ |r|  }{ } & &  $|BDD| = 1.9e+07$ & $|BDD| = 1.1e+08$ & 
\\

  	  \cline{2-5}
\multicolumn{1}{ |r|  }{ } & \multicolumn{1}{ |r||  }{\multirow{2}{*}{6v} } & 12.0 s & 
\multicolumn{1}{ |c|  }{\multirow{3}{*}{NA} } \\
  	  \multicolumn{1}{ |c|  }{ } &  &  $|S| = 3.3e+09$ & 
\\
  	  \multicolumn{1}{ |r|  }{ } & &  $|BDD| = 4.2e+07$ & 
\\

  	  \cline{2-4}
\multicolumn{1}{ |r|  }{ } & \multicolumn{1}{ |r||  }{\multirow{2}{*}{7v} } & 38.8 s & 
\multicolumn{1}{ |c|  }{\multirow{3}{*}{NA} } \\
  	  \multicolumn{1}{ |c|  }{ } &  &  $|S| = 7.6e+10$ & 
\\
  	  \multicolumn{1}{ |r|  }{ } & &  $|BDD| = 5.9e_07$ & 
\\

  	  \cline{2-4}
\multicolumn{1}{ |r|  }{ } & \multicolumn{1}{ |r||  }{\multirow{2}{*}{8v} } & 82.0 s & 
\multicolumn{1}{ |c|  }{\multirow{3}{*}{NA} } \\
  	  \multicolumn{1}{ |c|  }{ } &  &  $|S| = 8.7e+12 $ & 
\\
  	  \multicolumn{1}{ |r|  }{ } & &  $|BDD| = 5.9e+07$ & 
\\

  	  \cline{2-4}
\multicolumn{1}{ |r|  }{ } & \multicolumn{1}{ |r||  }{\multirow{2}{*}{9v} } & 139.5 s & 
\multicolumn{1}{ |c|  }{\multirow{3}{*}{NA} } \\
  	  \multicolumn{1}{ |c|  }{ } &  &  $|S| = 2.6e+14$  & 
\\
  	  \multicolumn{1}{ |r|  }{ } & &  $|BDD| = 6.1e+07$ & 
\\
  	  \cline{1-4}

\end{tabular}
\caption{\small MCMAS statistics for coercion freeness, $\G_{count}$ and \texttt{-atlk 2} flag}
\label{tab:count}
\end{table}

\end{landscape}

We also ran the same tests for the anonymity property,
$\varphi_i^{c} = AG ( \bigwedge_{1 \leq j \leq nc} \neg K_{att} p_{ch_i = j} )$
The results are given in Tables \ref{tab:totCTL}, \ref{tab:lexCTL} and \ref{tab:countCTL}.
Note the same anomalies of smaller state space but larger BDD size and running time between 
some $\G_{lex}$ models and the respective $\G_{tot}$ models.

\begin{table}[th]
\footnotesize
\centering
\begin{tabular}{r | c || c | c | c |}
	  \cline{3-5}
	  \multicolumn{2}{c|}{} & \multicolumn{3}{c|}{\# candidates} \\
	  \cline{3-5}
	  \multicolumn{2}{c|}{} & 2c & 3c & 4c \\
 	  \cline{3-5} \cline{1-5}
 	  \multicolumn{1}{ |r|  }{\multirow{4}{*}{\rotatebox{90}{\# voters}} } & \multicolumn{1}{ |r||  }{\multirow{2}{*}{2v} } & 1.1 s & 9.7 s & 3200.2 s \\
 	  \multicolumn{1}{ |r|  }{ } &  & $|S| \approx 5.9e+06$  & $|S| \approx 3.3e+08$ & $|S| \approx 8.8+09$ \\
 	  \multicolumn{1}{ |r|  }{ } &  & $|BDD| \approx 1.2e+07$  & $|BDD| \approx 4.5e+07  	$ & $|BDD| \approx 2.8e+08$ \\

 	  \cline{2-5}
 	  \multicolumn{1}{ |r|  }{ } & \multicolumn{1}{ |r||  }{\multirow{2}{*}{3v} } & 14.0 s & 6455.9 s & \multicolumn{1}{ |c|  }{\multirow{3}{*}{NA} } \\
 	  \multicolumn{1}{ |r|  }{ } &  &  $|S| \approx 2.9e+10$ & $|S| \approx 2.4e+13$ & \\
 	  \multicolumn{1}{ |r|  }{ } &  & $|BDD| \approx 5.4e+07$  & $|BDD| \approx 2.4e+08$ &  \\

 	  \cline{2-5}
 	  \multicolumn{1}{ |r|  }{ } & \multicolumn{1}{ |r||  }{\multirow{2}{*}{4v} } & 440.0 s & \multicolumn{1}{ |c|  }{\multirow{3}{*}{NA} } \\
 	  \multicolumn{1}{ |r|  }{ } &  &  $|S| \approx 1.7e+14$ & \\
 	  \multicolumn{1}{ |r|  }{ } &  & $|BDD| \approx 1.1e+08$  & \\
 	  \cline{1-4}
\end{tabular}
\caption{\small MCMAS statistics for anonymity checking on $\G_{tot}$}
\label{tab:totCTL}
\end{table}

\begin{table}[th]
\footnotesize
\centering
\begin{tabular}{r | c || c | c | c | c |}
	  \cline{3-6}
	  \multicolumn{2}{c|}{} & \multicolumn{3}{c|}{\# candidates} \\
	  \cline{3-6}
	  \multicolumn{2}{c|}{} & 2c & 3c & 4c \\
  	  \cline{3-6} \cline{1-6}
  	  \multicolumn{1}{ |r|  }{\multirow{4}{*}{\rotatebox{90}{\# voters}} } & \multicolumn{1}{ |r||  }{\multirow{2}{*}{2v} } & 0.9 s & 18.9 s & NA   \\
  	  \multicolumn{1}{ |r|  }{ } &  &  $|S| \approx 2.9e+05$ & $|S| \approx 1.1e+07$ &  \\
  	  \multicolumn{1}{ |r|  }{ } &  & $|BDD| \approx 1.3e+07  	$  & $|BDD| \approx 4.9e+07$ & \\

  	  \cline{2-6}
  	  \multicolumn{1}{ |r|  }{ } & \multicolumn{1}{ |r||  }{\multirow{2}{*}{3v} } & 8.3 s & 7100.6 s & \multicolumn{1}{ |c|  }{\multirow{3}{*}{NA} } \\
  	  \multicolumn{1}{ |c|  }{ } &  &  $|S| \approx 3.1e+08$ &  $|S| \approx 1.5e+11$ & \\
  	  \multicolumn{1}{ |r|  }{ } &  & $|BDD| \approx 3.9e+07  	$  & $|BDD| \approx 3.0e+08$ &  \\

  	  \cline{2-5}
  	  \multicolumn{1}{ |r|  }{ } & \multicolumn{1}{ |r||  }{\multirow{2}{*}{4v} } & 323.3 s & \multicolumn{1}{ |c|  }{\multirow{3}{*}{NA} } \\
  	  \multicolumn{1}{ |c|  }{ } &  &  $|S| \approx 4.0e+11$ &  \\
  	  \multicolumn{1}{ |r|  }{ } &  & $|BDD| \approx 6.5e+07 	$  & \\
  	  \cline{1-4}
\end{tabular}
\caption{\small MCMAS statistics for anonymity checking on $\G_{lex}$}
\label{tab:lexCTL}
\end{table}

\begin{landscape}
  \begin{table}[t]
\tiny
\centering
\begin{tabular}{r | c || c | c | c | c | c | c | c |}
	  \cline{3-8}
\multicolumn{2}{c|}{} & \multicolumn{7}{c|}{\# candidates} \\
	  \cline{3-8}
\multicolumn{2}{c|}{} & 2c & 3c & 4c & 5c & 6c & 7c 
\\
  	  \cline{3-8} \cline{1-6}
\multicolumn{1}{ |r|  }{\multirow{18}{*}{\rotatebox{90}{\# voters}} } & \multicolumn{1}{ |r||  }{\multirow{2}{*}{2v} } & 0.1 s & 0.4 s & 1.6 s & 7.6 s & 103.6 s & 
\multicolumn{1}{ |c|  }{\multirow{3}{*}{NA} } 
  	  \\
  	  \multicolumn{1}{ |r|  }{ } &  &  $|S| = 6.3e+03$ & $|S|= 1.6e+05$ & $|S| = 4.6e+06$  & $|S| = 1.2e+08$ & $|S| = 2.9e+09$ & 
\\
  	  \multicolumn{1}{ |r|  }{ } & &  $|BDD| = 1.0e+07$ & $|BDD| = 1.2e+07$ & $|BDD| = 1.9e+07$ & $|BDD| = 5.1e+07$ & $|BDD| = 9.3e+07$ & 
\\
  	
  	  \cline{2-8}
\multicolumn{1}{ |r|  }{ } & \multicolumn{1}{ |r||  }{\multirow{2}{*}{3v} } & 0.5 s & 4.1 s & 114.2 s & 
  	  5617.9 
& \multicolumn{1}{ |c|  }{\multirow{3}{*}{NA} }  
\\
  	  \multicolumn{1}{ |c|  }{ } &  &  $|S| = 9.3e+04$ & $|S| = 6.2e+06$ & $|S| = 4.2e+08$ & 
  	  $|S| = 3.2e+10$ & 
\\
  	  \multicolumn{1}{ |r|  }{ } & &  $|BDD| = 1.2e+07$ & $|BDD| = 2.9e+07$ & $|BDD| = 7.8e+07$ & 
  	  $|BDD| = 4.2e+08$ & 
\\

\cline{2-7}
  	  \multicolumn{1}{ |r|  }{ } & \multicolumn{1}{ |r||  }{\multirow{2}{*}{4v} } & 1.3 s & 29.5 s & 898.8 s & 
\multicolumn{1}{ |c|  }{\multirow{3}{*}{NA} } 
  	  \\
  	  \multicolumn{1}{ |c|  }{ } &  &  $|S| = 7.8e+06$ & $|S| = 3.6e+07$ & $|S| = 1.6+12$ & 
\\
  	  \multicolumn{1}{ |r|  }{ } & &  $|BDD| = 1.6e+07$ & $|BDD| = 5.4e+07$ & $|BDD| = 2.0e+08$ & 
\\

  	  \cline{2-6}
\multicolumn{1}{ |r|  }{ } & \multicolumn{1}{ |r||  }{\multirow{2}{*}{5v} } & 5.7 s & 336.9 s & 
\multicolumn{1}{ |c|  }{\multirow{3}{*}{NA} }  
  	  \\
  	  \multicolumn{1}{ |c|  }{ } &  &  $|S| = 1.7e+08$ & $|S| = 2.3e+11$ & 
\\
  	  \multicolumn{1}{ |r|  }{ } & &  $|BDD| = 1.9e+07$ & $|BDD| = 1.1e+08$ & 
\\

  	  \cline{2-5}
\multicolumn{1}{ |r|  }{ } & \multicolumn{1}{ |r||  }{\multirow{2}{*}{6v} } & 12.1  s & 
\multicolumn{1}{ |c|  }{\multirow{3}{*}{NA} } \\
  	  \multicolumn{1}{ |c|  }{ } &  &  $|S| = 3.2e+09$ & 
\\
  	  \multicolumn{1}{ |r|  }{ } & &  $|BDD| = 4.2e+07$ & 
\\

  	  \cline{2-4}
\multicolumn{1}{ |r|  }{ } & \multicolumn{1}{ |r||  }{\multirow{2}{*}{7v} } & 42.6 s & 
\multicolumn{1}{ |c|  }{\multirow{3}{*}{NA} } \\
  	  \multicolumn{1}{ |c|  }{ } &  &  $|S| = 7.6e+10$ & 
\\
  	  \multicolumn{1}{ |r|  }{ } & &  $|BDD| = 5.3e+07$ & 
\\

  	  \cline{2-4}
\multicolumn{1}{ |r|  }{ } & \multicolumn{1}{ |r||  }{\multirow{2}{*}{8v} } & 73.2 s &
\multicolumn{1}{ |c|  }{\multirow{3}{*}{NA} } \\
  	  \multicolumn{1}{ |c|  }{ } &  &  $|S| = 8.7e+12$ & 
\\
  	  \multicolumn{1}{ |r|  }{ } & &  $|BDD| = 5.9e+07$ & 
\\

  	  \cline{2-4}
\multicolumn{1}{ |r|  }{ } & \multicolumn{1}{ |r||  }{\multirow{2}{*}{9v} } & 140.0 s & 
\multicolumn{1}{ |c|  }{\multirow{3}{*}{NA} } \\
  	  \multicolumn{1}{ |c|  }{ } &  &  $|S| = 2.6e+14$  & 
\\
  	  \multicolumn{1}{ |r|  }{ } & &  $|BDD| = 6.0e+07$ & 
\\
  	  \cline{1-4}
\end{tabular}
\caption{\small MCMAS statistics for anonymity checking on $\G_{count}$}
\label{tab:countCTL}
\end{table}

\end{landscape}

\section{Conclusions} \label{conc}

In this paper we advance the state-of-the-art in model theory and verification for
the strategy logic ATL$^*$ under imperfect information and imperfect
recall. Specifically, we introduce a novel notion of (bi)simulation
on imperfect information concurrent game structures
that preserves the interpretation of formulas in ATL$^*$ both under
the objective and subjective variants of the semantics
(Theorem~\ref{thm-bisim}).  Then, we apply this theoretical result
to the verification of the \ThreeBallot voting system, a relevant
voting protocol without cryptography.  We show how the
\ThreeBallot protocol can be captured within the framework of iCGS, and then provide successive abstractions that are provably bisimilar to the original iCGS for the \ThreeBallot system.
In particular, we have been able to model check the ``smaller'' bisimilar
reductions of the \ThreeBallot model, and then transfer the result
to the original model in virtue of Theorem~\ref{thm-bisim}. As
reported in the experimental results, the gains in terms of both time
and memory resources are significant.

{\bf Future Work.}  We envisage several extensions of the present
contribution. First, it is of interest to develop bisimulations for
strategic properties of agents with perfect recall and bounded recall, as in many application domains
agents do have some memory of past states and actions. Also for the
verification of voting protocols, it is key to extend ATL$^*$ with
epistemic modalities to express naturally properties of anonymity and
confidentiality. We remarked that individual knowledge is expressible
in the subjective semantics. However, no such result holds for the
objective interpretation, nor common knowledge happens to be
definable. Finally, we aim at automating and implementing the
procedures described in this paper in a model checking tool for
the formal verification of (electronic) voting protocols.
 
\medskip\noindent\textbf{Acknowledgements.}
F.~Belardinelli acknowledges the support of the ANR JCJC Project SVeDaS (ANR-16-CE40-0021).
W.~Jamroga acknowledges the support of the National Centre for Research and Development (NCBR),
Po\-land, under the project VoteVerif
(POL\-LUX\--IV/1/2016) and PolLux/FNR-CORE project STV (POL\-LUX\--VII/1/2019).

\bibliographystyle{elsarticle-num}
\newcommand{\hoek}[1]{}

\nocite{EijckO07,KustersTV11,ChadhaKS06,NeuhausserK07}

\end{document}